\documentclass{article}

\usepackage{PRIMEarxiv}

\usepackage[utf8]{inputenc} 
\usepackage[T1]{fontenc}    
\usepackage{hyperref}       
\usepackage{url}            
\usepackage{booktabs}       
\usepackage{amsfonts}       
\usepackage{nicefrac}       
\usepackage{microtype}      
\usepackage{lipsum}
\usepackage{fancyhdr}       
\usepackage{graphicx}       
\usepackage{tikz-cd}
\usepackage{amsmath,amssymb,amsthm,mathrsfs}
\theoremstyle{plain}
\newtheorem{theorem}{Theorem}
\newtheorem{lemma}[theorem]{Lemma}
\newtheorem{corollary}[theorem]{Corollary}
\newtheorem{example}[theorem]{Example}
\theoremstyle{definition}

\newtheorem{remark}[theorem]{Remark}
\DeclareMathOperator{\Ima}{Im}

\graphicspath{{media/}}     

\title{Intersection of linear and multi-twisted codes with applications
\thanks{\textit{This research was conducted at Université d'Artois, la Faculté des Sciences Jean Perrin, France, and was fully funded by the Science, Technology \& Innovation Funding Authority (STDF); International Cooperation Grants, project number 49294. Ramy Takieldin would like to express his deepest appreciation to Faculté des Sciences Jean Perrin for their hospitality and providing a fruitful research environment.}} }

\author{
  Ramy Takieldin \\
  Faculty of Engineering, Ain Shams University, Cairo, Egypt\\
  Egypt University of Informatics, New Capital, Cairo, Egypt\\
  \texttt{ramy.farouk@eng.asu.edu.eg} \\
   \AND
   André Leroy \\
Univ.\ Artois, UR~2462, Laboratoire de Math\'ematique de Lens (LML), F-62300 Lens, France\\
   \texttt{andre.leroy@univ-artois.fr} \\
   \texttt{Andre.Leroy@math.cnrs.fr} \\
}

\begin{document}
\maketitle
\sloppy

\begin{abstract}
In this paper, we derive a formula for constructing a generator matrix for the intersection of any pair of linear codes over a finite field. Consequently, we establish a condition under which a linear code has a trivial intersection with another linear code (or its Galois dual). Furthermore, we provide a condition for reversibility and propose a generator matrix formula for the largest reversible subcode of any linear code. We then focus on the comprehensive class of multi-twisted (MT) codes, which are naturally and more effectively represented using generator polynomial matrices (GPMs). We prove that the reversed code of an MT code remains MT and derive an explicit formula for its GPM. Additionally, we examine the intersection of a pair of MT codes, possibly with different shift constants, and demonstrate that this intersection is not necessarily MT. However, when the intersection admits an MT structure, we propose the corresponding shift constants. We also establish a GPM formula for the intersection of a pair of MT codes with the same shift constants. This result enables us to derive a GPM formula for the intersection of an MT code and the Galois dual of another MT code. Finally, we examine conditions for various properties on MT codes. Perhaps most importantly, the necessary and sufficient conditions for an MT code to be Galois self-orthogonal, Galois dual-containing, Galois linear complementary dual (LCD), or reversible.
\end{abstract}

\keywords{Codes intersection \and Reversible \and Self-orthogonal \and Dual-containing \and LCD}

MSC: 94B05, 94B15, 11T71, 15A24

\section{Introduction}
\label{Sec-Intro}
Throughout this paper, $\mathbb{F}_q$ denotes a finite field of order $q$, and $\mathcal{C}$ represents a linear code of length $n$ over $\mathbb{F}_q$. The Euclidean dual of $\mathcal{C}$ is denoted by $\mathcal{C}^\perp$, whereas $\mathcal{C}^{\perp_\kappa}$ denotes the $\kappa$-Galois dual, which is defined using the Galois inner product introduced in \cite{FanZhang2017}. The Euclidean (or $\kappa$-Galois, respectively) hull of $\mathcal{C}$ is the intersection of $\mathcal{C}$ with its Euclidean (or $\kappa$-Galois, respectively) dual. Studying the Galois hulls of linear codes is important due to their applications in cryptography and quantum error-correcting codes construction \cite{CarletGuilley2016, LiuYuHu2019}. Furthermore, the dimension of the Galois hull of $\mathcal{C}$ determines whether $\mathcal{C}$ is Galois self-orthogonal, Galois dual-containing, or Galois LCD. In particular, $\mathcal{C}$ is $\kappa$-Galois LCD if $\mathcal{C} \oplus \mathcal{C}^{\perp_\kappa} = \mathbb{F}_q^n$. This concept extends to a linear complementary pair (LCP) of codes $\mathcal{C}_1$ and $\mathcal{C}_2$ if $\mathcal{C}_1 \oplus \mathcal{C}_2 = \mathbb{F}_q^n$. Furthermore, the concept of linear $\delta$-intersection pair of codes generalized LCP codes, where the pair $\mathcal{C}_1$ and $\mathcal{C}_2$ is a linear $\delta$-intersection if their intersection has dimension $\delta$. In \cite{GuendaGulliverJitman2020}, $\delta$-intersection pair codes are employed to construct good entanglement-assisted quantum error-correcting codes. In addition, a formula for determining the dimension of the intersection of any pair of linear codes was proposed in \cite{GuendaGulliverJitman2020}. Since prior works have only established the dimension of such intersection, a natural extension is to explicitly determine a generator matrix for the intersection of any pair of linear codes. This constitutes the first objective of this paper. In applications such as cryptography and the construction of quantum error-correcting codes, the properties of being self-orthogonal, dual-containing, and LCD are particularly relevant \cite{CarletGuilley2016, CaoCui2020}. However, in the context of DNA-based data storage and retrieval systems, the most crucial property of a linear code is being reversible \cite{OztasSiapYildiz2014, KimChoiLee2021}. A linear code is said to be reversible if it remains invariant under the reversal of the coordinates in each codeword.

A cyclic code is a linear code invariant under cyclic shifts of its codewords. Cyclic codes of length $n$ over $\mathbb{F}_q$ correspond bijectively to ideals in the quotient ring $\mathbb{F}_q[x]/\langle x^n-1\rangle$. Cyclic codes are significant not only because of their rich algebraic structure but also because of their practical applications, as they can be efficiently encoded and decoded using shift registers. Several generalizations of cyclic codes have been proposed in literature to achieve broader classes. Constacyclic codes provide a remarkable generalization of cyclic codes, where a $\lambda$-constacyclic code ($0 \neq \lambda \in \mathbb{F}_q$) of length $n$ over $\mathbb{F}_q$ is an ideal in the quotient ring $\mathbb{F}_q[x]/\langle x^n-\lambda\rangle$. Quasi-cyclic (QC) codes and $\lambda$-quasi-twisted (QT) codes, as discussed in \cite{Jia2012}, generalize cyclic and $\lambda$-constacyclic codes, respectively. QC and QT codes have gained importance as they were proven to be asymptotically good \cite{ZhuShi2022}. In \cite{TakiEldin2024, TakiEldin2025}, QC and $\lambda$-QT codes over $\mathbb{F}_q$ of index $\ell$ and co-index $m$ are in one-to-one correspondence with $\mathbb{F}_q[x]$-submodules of $\left(\mathbb{F}_q[x]\right)^\ell$ containing the submodule $\left(\langle x^m-1\rangle\right)^\ell$ and $\left(\langle x^m-\lambda\rangle\right)^\ell$, respectively. Multi-twisted (MT) codes were introduced in \cite{AydinHalilovic2017} as a comprehensive class of linear codes that includes cyclic, constacyclic, QC, and QT codes as subclasses. All these codes are linear and can be described by generator matrices, but they are more effectively represented by polynomials. Specifically, cyclic and constacyclic codes are identified by generator polynomials, while QC, QT, and MT codes are identified by generator polynomial matrices (GPMs) \cite{TakiEldin2024}. In the literature, many properties for these codes have been associated to their generator polynomials and GPMs. For instance, Massey demonstrated in \cite{Massey1964} that a cyclic code is LCD if and only if it is reversible. Similarly, \cite{TakiEldinMatsui2022} examined the relation between reversibility and self-orthogonality in terms of GPMs of QC codes. However, the $\kappa$-Galois duals and $\kappa$-Galois hulls of MT codes have been investigated in \cite{TakiEldin2023} and \cite{TakiEldinSole2025}, respectively. In \cite{HossainBandi2023}, a characterization of $\delta$-intersection pairs of cyclic and QC codes was presented using their generator polynomials. Given these results, it makes sense to examine the intersection of pairs of MT codes in terms of their GPMs. This is the second objective we have set for this paper.

In this paper, we not only determine the dimension of the intersection of a pair of linear or MT codes but also explicitly find generators for this intersection. Specifically, in Section \ref{Sec-Linear-Codes}, we derive a formula for the generator matrix of the intersection of any pair of linear codes $\mathcal{C}_1$ and $\mathcal{C}_2$ over $\mathbb{F}_q$, based on their generator matrices. Then we extend the result to determine a generator matrix for $\mathcal{C}_1^{\perp_\kappa} \cap \mathcal{C}_2$. Consequently, we derive the condition under which a pair of linear codes intersects trivially, i.e., $\mathcal{C}_1 \cap \mathcal{C}_2 = \{\mathbf{0}\}$. In addition, we establish necessary and sufficient conditions for a linear code $\mathcal{C}$ to be reversible or to intersect trivially with its reversed code $\mathcal{R}$, which is obtained by reversing the coordinates of each codeword in $\mathcal{C}$. Moreover, for any linear code that is not reversible, we derive a generator matrix for its largest reversible subcode.

In Section \ref{Sec-MT-Reversed-Code}, we focus on reversibility within the class of MT codes and its characterization by GPMs. We show that the reversed code of an MT code remains MT, but with possibly different shift constants and block lengths. Thus, we present a method for constructing a GPM for the reversed code of any MT code. This enables establishing a necessary and sufficient condition for the reversibility of MT codes in Theorem \ref{Th-SO-DC-R}.

In Section \ref{Sec-Intersection-MT-Codes}, we examine the intersection of a pair of MT codes that have the same block lengths. A counterexample was presented in \cite{LiuLiu2022} to negate the incorrect claim made in \cite{Dinh2014}, which states that the intersection of a pair of constacyclic codes with different shift constants remains constacyclic. We begin with Example \ref{Code3-Code4-R4}, which illustrates that this claim is also incorrect for MT codes with index $\ell \geq 2$. However, Example \ref{Code3-Code4-R5} shows that the intersection of a pair of MT codes, $\mathcal{C}_1$ and $\mathcal{C}_2$, with different shift constants may have an MT structure. When the intersection admits an MT structure, we identify its shift constants in Theorem \ref{Th-CondOfInter}. This identification depends on the minimum distances of $\mathcal{C}_1$ and $\mathcal{C}_2$. In Theorem \ref{Th-Intersection}, we prove a formula for computing a GPM for the intersection of a pair of MT codes. We assume that the two MT codes have the same shift constants to ensure that their intersection is MT, without imposing any restrictions on their minimum distance. Later, we generalize this result by determining a GPM for $\mathcal{C}_1^{\perp_\kappa} \cap \mathcal{C}_2$, where the $\kappa$-Galois dual $\mathcal{C}_1^{\perp_\kappa}$ is MT whenever $\mathcal{C}_1$ is MT.

We dedicate Section \ref{Sec-applications} to exploring applications of the preceding theoretical results. First, we establish necessary and sufficient conditions for one MT code to contain another, be contained in another, or trivially intersect another. Then, in Theorem \ref{Th-SO-DC-R}, we provide necessary and sufficient conditions for an MT code to be Galois self-orthogonal, Galois dual-containing, Galois LCD, or reversible. Unlike similar conditions in the literature, our conditions are both necessary and sufficient, depend solely on the GPM of the code, and do not require decomposing the code as a direct sum of shorter codes over different field extensions.

The subsequent sections are arranged as follows: Section \ref{Sec-Linear-Codes} investigates the intersection of any pair of linear codes. Section \ref{Sec-MT-Reversed-Code} analyzes the reversed code of any MT code. Section \ref{Sec-Intersection-MT-Codes} examines the intersection of a pair of MT codes and constructs a GPM for this intersection. Section \ref{Sec-applications} explores the necessary and sufficient conditions for certain properties of MT codes. Finally, the study is concluded in Section \ref{conclusion}.

\section{Intersection of linear codes}
\label{Sec-Linear-Codes}
The main objective of this section is to determine a generator matrix for the intersection of any pair of linear codes. While this result is significant in its own, its main significance is found in its subsequent application. For $i=1,2$, let $\mathcal{C}_i$ be a linear code of length $n$ over $\mathbb{F}_q$, with a generator matrix $G_i$ and a parity check matrix $H_i$. If $\mathcal{C}_i$ has dimension $k_i$, then $G_i$ is of size $k_i \times n$, while $H_i$ has size $(n-k_i) \times n$. The Euclidean dual $\mathcal{C}_i^\perp$ is the linear code of dimension $n-k_i$ generated by $H_i$. It follows directly that $G_i H_i^T = \mathbf{0}$ and $H_i G_i^T = \mathbf{0}$, where the transpose operation is denoted by $^T$ throughout the paper, and $\mathbf{0}$ consistently represents the zero matrix or vector of the appropriate size. A generator matrix formula for the intersection of any pair of linear codes can now be proved.

\begin{theorem}
\label{Th-Linear}
For $i = 1, 2$, let $\mathcal{C}_i$ be a linear code of length $n$ over $\mathbb{F}_q$ with dimension $k_i$, generator matrix $G_i$, and parity check matrix $H_i$. Define $\mathcal{Q}$ as the linear code of length $k_2$ over $\mathbb{F}_q$ generated by $H_1 G_2^T$, and let $P$ be a parity check matrix for $\mathcal{Q}$. Then, $P G_2$ is a generator matrix for the intersection $\mathcal{C}_1 \cap \mathcal{C}_2$.
\end{theorem}
\begin{proof}
Let $\mathcal{C}_{\text{int}}$ be the linear code of length $n$ over $\mathbb{F}_q$ generated by $PG_2$. Since $H_1 \left(P G_2\right)^T = H_1 G_2^T P^T = \mathbf{0}$, it follows that $\mathcal{C}_{\text{int}}\subseteq \mathcal{C}_1$. Furthermore, since the rows of $P G_2$ are linear combinations of the rows of $G_2$, we conclude that $\mathcal{C}_{\text{int}} \subseteq \mathcal{C}_1 \cap \mathcal{C}_2$. 

Let $\delta$ denote the row rank of $P$. We now present the following diagram of linear transformations:
\[
\begin{tikzcd}
\mathbb{F}_q^{\delta} \arrow[r, hook, "p"] & \mathbb{F}_q^{k_2} \arrow[r,equals, "g"] & \mathcal{C}_2  \arrow[r,"h"] & \mathbb{F}_q^{n-k_1}
  \end{tikzcd}
\]
where $p: a \mapsto a P$ is injective, $g: a \mapsto a G_2$ is bijective, and $h: a \mapsto a H_1^T$ has kernel $\mathcal{C}_1\cap \mathcal{C}2$. We abbreviate the dimension by $\dim$, the image by $\Ima$, and the kernel by $\ker$. Since $P$ is the parity check matrix of $\mathcal{Q}$, it follows that $\Ima (p)=\ker(h \circ g)$. Thus,
$$\dim(\mathcal{C}_{\text{int}})= \dim(\Ima (g\circ p))= \dim(\Ima (p))= \dim(\ker(h \circ g))= \dim(\ker(h))= \dim(\mathcal{C}_1\cap \mathcal{C}2).$$
Consequently, $\mathcal{C}_{\text{int}}=\mathcal{C}_1\cap \mathcal{C}2$.
\end{proof}

The following example illustrates the application of Theorem \ref{Th-Linear}. The particular codes defined in this example will be used frequently in other examples throughout the paper. An $[n,k,d]$ code refers to a linear code of length $n$, dimension $k$, and minimum Hamming distance $d$.
\begin{example}
\label{Code1-Code2}
In all examples presented in this paper, we define $\mathcal{C}_1$ as the $[8,6,2]$ code over $\mathbb{F}_4$ with generator and parity check matrices
\begin{equation*}
G_1=\begin{pmatrix}
 1 &0 &0 &0 &0 &\omega &0 &1\\
 0 &1 &0 &0 &0&\omega^2 &0 &1\\
 0 &0 &1 &0 &0 &1 &0 &0\\
 0 &0 &0 &1 &0 &\omega &0 &1\\
 0 &0 &0 &0 &1&\omega^2 &0 &1\\
 0 &0 &0 &0 &0 &0 &1 &\omega
\end{pmatrix} \quad \text{and}\quad 
H_1=\begin{pmatrix}
1& 0& 1& 1& 0& 1& 1& \omega^2\\
0& 1& 1& 0& 1& 1& \omega^2& \omega
\end{pmatrix},
\end{equation*}
respectively, where $\omega^2+\omega+1=0$. Then, $\mathcal{C}_1^\perp$ is the $[8,2,6]$ code generated by $H_1$. In all examples presented in this paper, we define $\mathcal{C}_2$ as the $[8,3,5]$ code  over $\mathbb{F}_4$ with generator matrix
\begin{equation*}
G_2=\begin{pmatrix}
1& 0& 0 &\omega^2 &\omega^2& 1& 1& 0\\
0& 1& 0 &\omega^2& 0& \omega& 1& 1\\
0& 0& 1& 1& \omega& \omega& 0& 1
\end{pmatrix}.
\end{equation*}

To apply Theorem \ref{Th-Linear} in determining a generator matrix for $\mathcal{C}_1 \cap \mathcal{C}_2$, we first consider the $[3,1,3]$ code $\mathcal{Q}$ over $\mathbb{F}_4$, generated by the rows of 
\begin{equation*}
H_1 G_2^T=\begin{pmatrix}
\omega & \omega^2  & 1\\
1  & \omega &\omega^2
\end{pmatrix}.
\end{equation*}
A parity check matrix for $\mathcal{Q}$ is given by 
\begin{equation*}
P=\begin{pmatrix}
1  & 0  & \omega\\
0  & 1 & \omega^2
\end{pmatrix}.
\end{equation*}
Then, 
\begin{equation*}
PG_2=\begin{pmatrix}
1  & 0  & \omega  & 1  & 0  & \omega  & 1  & \omega\\
0  & 1 & \omega^2  & 0  & 1 & \omega^2  & 1  & \omega
\end{pmatrix}
\end{equation*}
is a generator matrix for $\mathcal{C}_1\cap \mathcal{C}_2$. That is, $\mathcal{C}_1 \cap \mathcal{C}_2$ is an $[8,2,6]$ code.  
\hfill $\diamond$
\end{example}

The following is a necessary and sufficient condition for a pair of codes to intersect trivially, it follows as a direct consequence of Theorem \ref{Th-Linear}.  
\begin{corollary}
\label{Linear_LCP}
For $i = 1, 2$, let $\mathcal{C}_i$ be a linear code of length $n$ over $\mathbb{F}_q$ with dimension $k_i$, generator matrix $G_i$, and parity check matrix $H_i$. Then $\mathcal{C}_1 \cap \mathcal{C}_2=\{\mathbf{0}\}$ if and only if $\mathrm{rank}\left( H_1 G_2^T\right)=k_2$. 
\end{corollary}
\begin{proof}
This follows from Theorem \ref{Th-Linear} since
\begin{equation*}
\begin{split}
\mathrm{rank}\left( H_1 G_2^T\right)=\dim\left(\mathcal{Q}\right)=k_2- \dim\left(\mathcal{Q}^\perp\right)&=k_2-\mathrm{rank}\left(P\right)\\
&=k_2-\mathrm{rank}\left(P G_2\right)=k_2-\dim\left( \mathcal{C}_1 \cap \mathcal{C}_2\right).
\end{split}
\end{equation*}
\end{proof}

The $\kappa$-Galois dual of a linear code $\mathcal{C}$ of length $n$ over $\mathbb{F}_q$ is used as a generalization of the Euclidean dual $\mathcal{C}^\perp$. In particular, $\mathcal{C}^\perp$ corresponds to the $\kappa$-Galois dual with $\kappa = 0$. Hereinafter, let $q = p^e$ and $0 \leq \kappa < e$, where $p$ is a prime and $e$ and $\kappa$ are positive integers. Denote by $\sigma$ the Frobenius automorphism of $\mathbb{F}_q$, given by $\sigma(\alpha) = \alpha^p$ for all $\alpha \in \mathbb{F}_q$. The $\kappa$-Galois inner product is defined as
$$\langle \mathbf{u},\mathbf{v}\rangle_\kappa=\sum_{i=0}^n u_i \sigma^\kappa(v_i),$$
for any $\mathbf{u} = \left(u_1, \ldots, u_n\right) \in \mathbb{F}_q^n$ and $\mathbf{v} = \left(v_1, \ldots, v_n\right) \in \mathbb{F}_q^n$. Similarly, the $\kappa$-Galois dual of $\mathcal{C}$ is defined as
$$\mathcal{C}^{\perp_\kappa}=\left\{\mathbf{v}\in\mathbb{F}_q^n \text{ such that } \langle \mathbf{c},\mathbf{v}\rangle_\kappa =0 \quad \forall \mathbf{c}\in \mathcal{C} \right\}.$$
It is straightforward to verify that $\mathcal{C}^{\perp_\kappa} = \sigma^{e-\kappa}(\mathcal{C}^\perp)$, where $\sigma^{e-\kappa}$ is applied component-wise to each codeword. Accordingly, $\sigma^{e-\kappa}(H)$ serves as a generator matrix for $\mathcal{C}^{\perp_\kappa}$ whenever $H$ is a parity check matrix for $\mathcal{C}$, with $\sigma$ acting element-wise on the matrix. This leads to the following consequence of Theorem \ref{Th-Linear}, which determines a generator matrix for the intersection between a linear code and the Galois dual of another. In the special case where the two codes are identical, this yields a generator matrix for the Galois hull. This is important, as existing results in the literature are limited to determining the dimension of the Galois hull \cite{Cao2024}, without offering a method for constructing its generator matrix.

\begin{corollary}
\label{Th-LinearGaloisIntersection}
For $i=1, 2$, let $\mathcal{C}_i$ be a linear code of length $n$ over $\mathbb{F}_q$ with dimension $k_i$, generator matrix $G_i$, and parity check matrix $H_i$. Let $\mathcal{Q}$ be the linear code of length $k_2$ over $\mathbb{F}_q$ generated by $\sigma^{e-\kappa}\left(G_1\right) G_2^T$, and let $P$ be a parity check matrix of $\mathcal{Q}$. Then $PG_2$ is a generator matrix for $\mathcal{C}_1^{\perp_\kappa} \cap \mathcal{C}_2$. 

In particular,  $PG$ is a generator matrix for the $\kappa$-Galois hull of the linear code with generator matrix $G$, where $P$ in this case is a parity check matrix of the linear code generated by $\sigma^{e-\kappa}\left(G\right) G^T$.
\end{corollary}
\begin{proof}
The result follows from Theorem \ref{Th-Linear} by using $\sigma^{e-\kappa}\left(H_1\right)$ and $\sigma^{e-\kappa}\left(G_1\right)$ as a generator matrix and a parity check matrix for the $\kappa$-Galois dual $\mathcal{C}_1^{\perp_\kappa}$, respectively. The second assertion follows by letting $\mathcal{C}_1 = \mathcal{C}_2$.
\end{proof}

\begin{example}
\label{Code1-Code2-R4}
Consider the linear codes $\mathcal{C}_1$ and $\mathcal{C}_2$ given in Example \ref{Code1-Code2}. In this example, we apply Corollary \ref{Th-LinearGaloisIntersection} to determine the intersection between the $1$-Galois dual of $\mathcal{C}_1$ and $\mathcal{C}_2$. To this end, we let $\mathcal{Q}$ be the linear code of length $k_2=3$ generated by
\begin{equation*}
\sigma\left(G_1\right) G_2^T=\begin{pmatrix}
 \omega  &  0  &  0\\
 \omega & \omega^2  &   \omega\\
 1  &  \omega&\omega^2\\
 0 & \omega^2  &   1\\
 1  &   \omega   &  0\\
 1  &   \omega & \omega^2
\end{pmatrix}.
\end{equation*}
Clearly, $\mathcal{Q}=\mathbb{F}_4^3$ since $\mathrm{rank}\left(\sigma\left(G_1\right) G_2^T\right) = 3$. Since the zero matrix is a parity check matrix of $\mathcal{Q}$, Corollary \ref{Th-LinearGaloisIntersection} implies that $\mathcal{C}_1^{\perp_1} \cap \mathcal{C}_2=\{\mathbf{0}\}$. 
\hfill $\diamond$
\end{example}

We now turn our attention to the reversibility of linear codes. Hereinafter, we denote the $n \times n$ identity matrix by $\mathbf{I}_n$ and the $n \times n$ backward identity matrix by $J_n$. Let $\mathcal{C}$ be a linear code of length $n$ over $\mathbb{F}_q$ with generator matrix $G$. The reversed code $\mathcal{R}$ of $\mathcal{C}$ is defined as the linear code obtained by reversing the coordinates of each codeword in $\mathcal{C}$. Clearly, $\mathcal{C}$ and $\mathcal{R}$ have the same length, dimension, and minimum distance. In addition, $G J_n$ and $H J_n$ are generator and parity check matrices for $\mathcal{R}$, respectively. We call $\mathcal{C}$ reversible if $\mathcal{R} = \mathcal{C}$. The following result provides a condition under which a linear code is reversible. It also determines when a linear code intersects its reversed code trivially. Moreover, it establishes a generator matrix for the largest reversible subcode contained in any linear code. These results will be proved using Theorem \ref{Th-Linear}.

\begin{theorem}
\label{reversible-subcode}
Let $\mathcal{C}$ be a linear code of length $n$ over $\mathbb{F}_q$ with dimension $k$, generator matrix $G$, and parity check matrix $H$. Then
\begin{enumerate}
\item $\mathcal{C}$ is reversible if and only if $H J_n G^T = \mathbf{0}$.
\item If $H J_n G^T \ne \mathbf{0}$, then $PG$ is a generator matrix for the largest reversible subcode of $\mathcal{C}$, where $P$ is a parity check matrix of the linear code of length $k$ generated by $H J_n G^T$. 
\item $\mathcal{C}$ contains no nontrivial reversible subcode if and only if $\mathrm{rank}\left(H J_n G^T\right)=k$.
\end{enumerate}
\end{theorem}
\begin{proof}
Using the same notation as in Theorem \ref{Th-Linear}, setting $\mathcal{C}_1 = \mathcal{R}$, the reversed code of $\mathcal{C}$, and $\mathcal{C}_2 = \mathcal{C}$, then $\mathcal{Q}$ is the code of length $k$ generated by $H J_n G^T$.
\begin{enumerate}
\item $H J_n G^T = \mathbf{0}$ if and only if $\mathcal{Q}$ is the zero code. However, $\mathcal{Q}$ is the zero code if and only if its parity check matrix $P$ is invertible. But $P$ is invertible if and only if $\mathcal{R} \cap \mathcal{C}=\mathcal{C}$, meaning that $\mathcal{C}$ is reversible.
\item If $H J_n G^T \ne \mathbf{0}$, then the intersection $\mathcal{R}\cap \mathcal{C}$ is generated by $PG$. We now show that $\mathcal{R}\cap \mathcal{C}$ is the largest reversible subcode of $\mathcal{C}$. Let $\mathcal{S}$ be an arbitrary reversible subcode of $\mathcal{C}$. Since $\mathcal{S} \subseteq \mathcal{C}$, taking the reversed codes of both sides implies that $\mathcal{S} \subseteq \mathcal{R}$. Consequently, for every reversible subcode $\mathcal{S}$ of $\mathcal{C}$, we have $\mathcal{S} \subseteq \mathcal{R}\cap \mathcal{C}$. Since $\mathcal{R}\cap \mathcal{C}$ is itself reversible, it follows that it is the largest reversible subcode of $\mathcal{C}$.
\item From the above assertion, $\mathcal{C}$ contains no nontrivial reversible subcode if and only if $\mathcal{C}$ intersects $\mathcal{R}$ trivially. The result then follows directly from Corollary \ref{Linear_LCP}.
\end{enumerate}
\end{proof}

\begin{example}
\label{Code3-Code4}
In this example, we define three linear codes, $\mathcal{C}_3$, $\mathcal{C}_4$, and $\mathcal{C}_5$. These codes will be referenced frequently in the examples throughout this paper. Let $\mathcal{C}_3$ be the $[9,3,6]$ code over $\mathbb{F}_3$ with generator and parity check matrices
\begin{equation*}
G_3=\begin{pmatrix}
1&0&0&1&1&2&1&1&0\\
0&1&0&2&1&1&0&1&1\\
0&0&1&1&2&1&1&0&1
\end{pmatrix} \quad \text{and} \quad
H_3=\begin{pmatrix}
1&0&0&0&0&0&1&1&2\\
0&1&0&0&0&0&2&1&1\\
0&0&1&0&0&0&1&2&1\\
0&0&0&1&0&0&0&2&2\\
0&0&0&0&1&0&2&0&2\\
0&0&0&0&0&1&2&2&0
\end{pmatrix},
\end{equation*}
respectively. Theorem \ref{reversible-subcode} asserts that $\mathcal{C}_3$ is not reversible since
\begin{equation*}
H_3 J_9 G_3^T =\begin{pmatrix}
2&2&2\\
2&2&2\\
2&2&2\\
1&0&1\\
0&1&1\\
1&1&0
\end{pmatrix}\ne \mathbf{0}.
\end{equation*}
In addition, $\mathcal{C}_3$ has no nontrivial reversible subcode because $\mathrm{rank}\left(H_3 J_9 G_3^T\right)=3$.

Our next code is the $[9,7,1]$ code $\mathcal{C}_4$ over $\mathbb{F}_3$ with generator and parity check matrices 
\begin{equation*}
G_4=\begin{pmatrix}
1&0&2&0&0&0&0&0&0\\
0&1&1&0&0&0&0&0&0\\
0&0&0&1&0&0&0&0&0\\
0&0&0&0&1&0&0&0&0\\
0&0&0&0&0&1&0&0&0\\
0&0&0&0&0&0&1&0&2\\
0&0&0&0&0&0&0&1&1
\end{pmatrix} \quad \text{and} \quad
H_4=\begin{pmatrix}
1&2&1&0&0&0&0&0&0\\
0&0&0&0&0&0&1&2&1
\end{pmatrix},
\end{equation*}
respectively. Theorem \ref{reversible-subcode} asserts that $\mathcal{C}_4$ is reversible since $H_4 J_9 G_4^T =\mathbf{0}$. To be used later in Example \ref{Code3-Code4-R4}, we obtain 
\begin{equation*}
\begin{pmatrix}
1&1&0&0&2&0&1&2&1
\end{pmatrix}
\end{equation*}
by Theorem \ref{Th-Linear} as a generator matrix for $\mathcal{C}_3 \cap \mathcal{C}_4$.

Another code that will be used in later examples is the $[9,6,3]$ code $\mathcal{C}_5$ over $\mathbb{F}_3$, with the following generator and parity check matrices 
\begin{equation*}
G_5=\begin{pmatrix}
1&0&0&0&0&0&1&0&2\\
0&1&0&0&0&0&1&1&0\\
0&0&1&0&0&0&0&1&1\\
0&0&0&1&0&0&1&0&1\\
0&0&0&0&1&0&2&1&0\\
0&0&0&0&0&1&0&2&1
\end{pmatrix} \quad \text{and} \quad
H_5=\begin{pmatrix}
1&0&2&0&2&0&1&2&2\\
0&1&1&0&1&2&0&2&0\\
0&0&0&1&2&1&1&2&1
\end{pmatrix}.
\end{equation*}
Theorem \ref{reversible-subcode} shows that $\mathcal{C}_5$ is not reversible since $H_5 J_9 G_5^T\ne\mathbf{0}$. The linear code of length $6$ generated by $H_5 J_9 G_5^T$ has a parity check matrix denoted by $P$, where
\begin{equation*}
H_5 J_9 G_5^T =\begin{pmatrix}
0&1&2&0&0&1\\
1&1&1&0&1&2\\
1&2&1&1&2&1
\end{pmatrix}\quad \text{and} \quad
P=\begin{pmatrix}
1&0&0&1&2&0\\
0&1&0&0&1&2\\
0&0&1&1&0&1
\end{pmatrix}.
\end{equation*}
Consequently, a generator matrix for the largest reversible subcode of $\mathcal{C}_5$ is given by
\begin{equation*}
P G_5 =\begin{pmatrix}
1&0&0&1&2&0&0&2&0\\
0&1&0&0&1&2&0&0&2\\
0&0&1&1&0&1&1&0&0
\end{pmatrix}.
\end{equation*}
\hfill $\diamond$
\end{example}

We conclude this section with the following remark, which plays an important role in generalizing certain results in Section \ref{Sec-applications}
\begin{remark}
\label{Rem_Chain}
In this remark, we present the standard form of a generator matrix for any linear code over a finite chain ring, as established in \cite[Theorem 2.12]{Dougherty2017}. Of particular interest is the case of a linear code $\mathcal{Q}$ over the finite chain ring $\mathbb{F}_q[x]/\langle p^f (x)\rangle$, where $\mathbb{F}_q[x]$ denotes the ring of polynomials over $\mathbb{F}_q$, $p(x) \in \mathbb{F}_q[x]$ is an irreducible polynomial, $f$ is a positive integer, and $\langle p^f (x)\rangle$ is the ideal generated by $p^f (x)$. It was shown in \cite{Dougherty2017} that $\mathcal{Q}$ admits a generator matrix of the form
\begin{equation}
\label{MatrixChainRing}
\begin{pmatrix}
\mathbf{I}_{r_0} & \star & \star & \star & \cdots & \cdots & \star \\
\mathbf{0} & p(x) \mathbf{I}_{r_1} & p(x) \star & p(x) \star & \cdots & \cdots & p(x) \star \\
\mathbf{0} & \mathbf{0} & p^2 (x) \mathbf{I}_{r_2} & p^2 (x) \star & \cdots & \cdots & p^2 (x) \star \\
\vdots & \vdots & \mathbf{0} & \ddots & \ddots & \ddots & \vdots \\
\vdots & \vdots & \vdots & \ddots & \ddots & \ddots & \vdots \\
\mathbf{0} & \mathbf{0} & \mathbf{0} & \cdots & \mathbf{0} & p^{f-1}(x) \mathbf{I}_{r_{f-1}} & p^{f-1}(x) \star
\end{pmatrix},
\end{equation}
where each $\star$ represents an arbitrary matrix with elements from $\mathbb{F}_q[x]/\langle p^f (x)\rangle$. A code with generator matrix of the form \eqref{MatrixChainRing} is said to have type $\left\{r_0, r_1, \ldots, r_{f-1}\right\}$. Corollary 2.2 in \cite{Dougherty2017} states that the size of $\mathcal{Q}$ is given by
\begin{equation}
\label{Eq_Chain}
|\mathcal{Q}|=q^{\deg\left(p(x)\right)\sum_{h=0}^{f-1}(f-h) r_h},
\end{equation}
where $\deg$ stands for the degree.
\hfill $\diamond$
\end{remark}

\section{MT codes and their reversed codes}
\label{Sec-MT-Reversed-Code}
From now on, we focus on the class of MT codes. We aim to specify the results of Section \ref{Sec-Linear-Codes} to MT codes. MT codes constitute a promising class as it includes several significant subclasses. For instance, a constacyclic code is an MT code with index $\ell = 1$, a generalized QC code corresponds to an MT code with unity shift constants, and a QT code is an MT code with equal shift constants and equal block lengths. Although all of these classes are linear and can be studied using generator matrices, their rich algebraic structures make them effectively studied using generator polynomials, or GPMs. Representing the results of Section \ref{Sec-Linear-Codes} in terms of GPMs makes them more effective for MT codes and their subclasses. In addition, deriving the theoretical results on MT codes makes it easy to derive corresponding results on any of the aforementioned subclases. Before we begin studying the reversed codes of MT codes, we begin this section by providing a brief description of the algebraic structure of MT codes, constructing their GPMs, and discussing some of the properties of GPMs. More details on these properties can be found in \cite{TakiEldin2024, TakiEldin2025, TakiEldin2023}.

For a given positive integer $\ell$, referred to as the code index, let $m_i$ be a positive integer and $\lambda_i\in \mathbb{F}_q$ be nonzero for each $i=1, 2, \ldots, \ell$. Set $\Lambda=\left(\lambda_1,\lambda_2, \ldots, \lambda_\ell\right)$ and define the code length $n$ as the sum of the block lengths $m_i$, i.e., $n=m_1+m_2+\cdots+m_\ell$. Each vector 
\begin{equation}
\label{random_vector}
\mathbf{a}=\left(a_{0,1},a_{1,1},\ldots, a_{m_1-1,1},a_{0,2},a_{1,2},\ldots, a_{m_2-1,2},\ldots, a_{0,\ell},a_{1,\ell},\ldots, a_{m_\ell-1,\ell} \right)\in\mathbb{F}_q^n
\end{equation}
can be expressed in a polynomial vector representation, denoted by $\mathbf{a}(x)$, which consists of $\ell$ components, where the $i$-th component is a polynomial of degree less than $m_i$ Specifically, 
\begin{equation}
\label{random_poly_vector}
\mathbf{a}(x)=\left(a_1(x), a_2(x),\ldots, a_\ell(x)\right),
\end{equation}
where $a_i(x)=a_{0,i}+a_{1,i}x+\cdots + a_{m_i-1,i}x^{m_i-1}$ is an element of the quotient ring $\mathbb{F}_q[x]/\langle x^{m_i}-\lambda_i\rangle$. This representation induces an $\mathbb{F}_q$-vector space isomorphism between $\mathbb{F}_q^n$ and the direct sum $\bigoplus_{i=1}^\ell \mathbb{F}_q[x]/\langle x^{m_i}-\lambda_i\rangle$. We naturally extend this isomorphism to an $\mathbb{F}_q[x]$-module isomorphism by endowing $\mathbb{F}_q^n$ with a module structure. Precisely, we make $\mathbb{F}_q^n$ into an $\mathbb{F}_q[x]$-module by defining the action $x(\mathbf{a}) = \mathcal{T}_\Lambda(\mathbf{a})$, where $\mathcal{T}_\Lambda$ is the linear transformation defined by
\begin{align*}
\mathcal{T}_\Lambda \left(\mathbf{a}\right)=\left(\lambda_1 a_{m_1-1,1},a_{0,1},a_{1,1},\ldots, a_{m_1-2,1},\lambda_2 a_{m_2-1,2}, a_{0,2},a_{1,2},\ldots, a_{m_2-2,2},\ldots, \qquad\qquad\qquad\right. \\ \left. \lambda_\ell a_{m_\ell-1,\ell}, a_{0,\ell},a_{1,\ell},\ldots, a_{m_\ell-2,\ell} \right)
\end{align*}
for every $\mathbf{a} \in\mathbb{F}_q^n$, as defined in \eqref{random_vector}. In other words, the map $\mathbf{a}\mapsto \mathbf{a}(x)$ not only establishes a one-to-one correspondence between vectors in $\mathbb{F}_q^n$ and polynomial vectors in $\bigoplus_{i=1}^\ell \mathbb{F}_q[x]/\langle x^{m_i}-\lambda_i\rangle$, but is also linear with respect to multiplication by elements of $\mathbb{F}_q[x]$.

A linear code $\mathcal{C}$ of length $n$ over $\mathbb{F}_q$ is called $\Lambda$-MT with block lengths $\left(m_1, m_2, \ldots, m_\ell\right)$ and shift constants $\Lambda$ if it is invariant under the transformation $\mathcal{T}_\Lambda$. That is, $\mathcal{C}$ is $\Lambda$-MT if $\mathcal{T}_\Lambda\left(\mathbf{c}\right)\in\mathcal{C}$ for every $\mathbf{c}\in\mathcal{C}$. By employing the $\mathbb{F}_q[x]$-module isomorphism discussed above, $\mathcal{C}$ can be viewed as an $\mathbb{F}_q[x]$-submodule of $\bigoplus_{i=1}^\ell \mathbb{F}_q[x]/\langle x^{m_i}-\lambda_i\rangle$, meaning that $a(x) \mathbf{c}(x) \in \mathcal{C}$ for any codeword $\mathbf{c}(x)\in \mathcal{C}$ written in the polynomial vector representation and any polynomial $a(x)\in\mathbb{F}_q[x]$. From the classical correspondence theorem of modules, any submodule $\mathcal{C}$ of $\bigoplus_{i=1}^\ell \mathbb{F}_q[x]/\langle x^{m_i}-\lambda_i\rangle$ corresponds to a submodule of $\left(\mathbb{F}_q[x]\right)^\ell$ that contains the submodule $\bigoplus_{i=1}^\ell \langle x^{m_i}-\lambda_i\rangle$. Throughout the paper, we will interchangeably regard $\mathcal{C}$ either as a subspace of $\mathbb{F}_q^n$ that is invariant under $\mathcal{T}_\Lambda$, or as an $\mathbb{F}_q[x]$-submodule of $\left(\mathbb{F}_q[x]\right)^\ell$ that contains $\bigoplus_{i=1}^\ell \langle x^{m_i}-\lambda_i\rangle$. Analogously, a codeword of $\mathcal{C}$ can be expressed as a vector $\mathbf{c}\in\mathbb{F}_q^n$ or as a polynomial vector $\mathbf{c}(x)\in\left(\mathbb{F}_q[x]\right)^\ell$. The second perspective is particularly useful in constructing a GPM for $\mathcal{C}$.

By regarding a $\Lambda$-MT code $\mathcal{C}$ as a submodule of $\left(\mathbb{F}_q[x]\right)^\ell$ containing the submodule $\bigoplus_{i=1}^\ell \langle x^{m_i}-\lambda_i\rangle$, a GPM for $\mathcal{C}$ is a matrix with polynomial entries whose rows generate $\mathcal{C}$. Throughout this paper, we use the term GPM to refer specifically to a polynomial matrix whose rows form a minimal generating set for $\mathcal{C}$. It was shown in \cite{TakiEldin2025} that $\mathcal{C}$ as a submodule of $\left(\mathbb{F}_q[x]\right)^\ell$ has rank $\ell$. Consequently, a GPM is of size $\ell \times \ell$, and its rows are $\mathbb{F}_q[x]$-linearly independent. However, any polynomial matrix with more than $\ell$ rows whose rows generate $\mathcal{C}$ can be reduced to a GPM, see Example \ref{Code1-Code2-R1}. In fact, GPMs for MT codes are similar to generator matrices for linear codes; a GPM acts as a generator matrix for $\mathcal{C}$ as a code of length $\ell$ over $\mathbb{F}_q[x]$. In addition, just as the row reduced echelon form determines a unique generator matrix for a linear code, the Hermite normal form defines a unique GPM for an MT code. A GPM in Hermite normal form is referred to as the reduced GPM. We denote a GPM by $\mathbf{G}$ and, more generally, use capital bold letters to represent polynomial matrices.

The requirement for $\mathcal{C}$ to include the submodule $\bigoplus_{i=1}^\ell \langle x^{m_i}-\lambda_i\rangle$ imposes a fundamental identity that any polynomial matrix generating $\mathcal{C}$ must fulfill. For instance, the rows of any GPM $\mathbf{G}$ are capable of generating a basis for $\bigoplus_{i=1}^\ell \langle x^{m_i}-\lambda_i\rangle$. Equivalently, there exists a polynomial matrix $\mathbf{A}$ such that
\begin{equation}
\label{Identical_Eq}
\mathbf{A}\mathbf{G}=\mathrm{diag}\left(x^{m_i}- \lambda_i \right),
\end{equation}
where $\mathrm{diag}\left(x^{m_i}- \lambda_i \right)$ denotes the $\ell \times \ell$ diagonal polynomial matrix with diagonal entries $x^{m_1}- \lambda_1, x^{m_2}- \lambda_2, \ldots, x^{m_\ell}- \lambda_\ell$. The identity given by \eqref{Identical_Eq} is referred to as the identical equation of $\mathbf{G}$. Conversely, any polynomial matrix that satisfies the identical equation for some $\mathbf{A}$ is a GPM for a $\Lambda$-MT code with block lengths $\left(m_1, m_2, \ldots, m_\ell\right)$. The polynomial matrix $\mathbf{A}$ plays a crucial role in constructing a GPM for the dual code of $\mathcal{C}$. Additionally, it determines the dimension of $\mathcal{C}$ as a subspace of $\mathbb{F}_q^n$ through the formula
\begin{equation}
\label{Dim_Eq_fromA}
\dim{\mathcal{C}}=\deg\left(\mathrm{det}\left( \mathbf{A}\right)\right),
\end{equation}
where $\mathrm{det}$ stands for the determinant. In the following example we show how a GPM can be constructed from a generator matrix of $\mathcal{C}$.

\begin{example}
\label{Code1-Code2-R1}
Consider the $[8,6,2]$ code $\mathcal{C}_1$ with generator matrix $G_1$, as presented in Example \ref{Code1-Code2}. Specifically, $\mathcal{C}_1$ is invariant under the transformation $\mathcal{T}_\Lambda$, where $\Lambda = (1, \omega)$, with $m_1 = 6$ and $m_2 = 2$. Consequently, $\mathcal{C}_1$ is a $(1,\omega)$-MT code with index $\ell = 2$ and block lengths $(6,2)$. A polynomial matrix generating $\mathcal{C}_1$ can be constructed using the polynomial vector representation of the rows of $G_1$, along with a basis for $\langle x^6-1\rangle\oplus \langle x^2-\omega\rangle$, and is given by:
\begin{equation*}
\begin{pmatrix}
1+\omega x^5 &  x \\
x+\omega^2 x^5 &  x \\
x^2+ x^5 &  0 \\
x^3+\omega x^5 &  x \\
x^4+\omega^2 x^5 &  x \\
0 &  1+\omega x \\
x^6-1 & 0\\
0 & x^2-\omega
\end{pmatrix}.
\end{equation*}
By reducing this polynomial matrix to its Hermite normal form, we obtain the reduced GPM for $\mathcal{C}_1$ as
\begin{equation*}
\mathbf{G}_1=\begin{pmatrix}
 \omega +x &  \omega \\
 0 & \omega^2+x
\end{pmatrix}.
\end{equation*}
However, the polynomial matrix $\mathbf{A}_1$ that satisfies the identical equation \eqref{Identical_Eq} for $\mathbf{G}_1$ is given by
\begin{equation*}
\mathbf{A}_1=\begin{pmatrix}
\omega^2+ \omega x + x^2 + \omega^2 x^3 + \omega x^4 +x^5  &  \omega+ \omega x + \omega x^3 +\omega x^4 \\
0 & \omega^2+x
\end{pmatrix}.
\end{equation*}
Observe that the dimension of $\mathcal{C}_1$ can be determined by \eqref{Dim_Eq_fromA}. In fact, $\deg\left(\mathrm{det}\left( \mathbf{A}_1\right)\right)=6$.

Similarly, consider the $[8,3,5]$ code $\mathcal{C}_2$ with generator matrix $G_2$, as described in Example \ref{Code1-Code2}. In fact, $\mathcal{C}_2$ is a $(1,\omega)$-MT code with index $\ell=2$ and block lengths $(6,2)$. A polynomial matrix generating $\mathcal{C}_2$ can be constructed using the polynomial vector representation of the rows of $G_2$, along with a basis for $\langle x^6-1\rangle\oplus \langle x^2-\omega\rangle$, and is given by:
\begin{equation*}
\begin{pmatrix}
1+\omega^2 x^3+\omega^2 x^4+ x^5 &  1 \\
x+\omega^2 x^3+\omega x^5 &  1+x \\
x^2+x^3+\omega x^4+\omega x^5 & x \\
x^6-1 & 0\\
0 & x^2-\omega
\end{pmatrix}.
\end{equation*}
By reducing this polynomial matrix to its Hermite normal form, we obtain the reduced GPM for $\mathcal{C}_2$ as
\begin{equation*}
\mathbf{G}_2=\begin{pmatrix}
 \omega^2+ \omega^2 x + x^2+x^3  & \omega x \\
 0 & \omega+x^2
\end{pmatrix}.
\end{equation*}
However, the polynomial matrix $\mathbf{A}_2$ that satisfies the identical equation for $\mathbf{G}_2$ is given by
\begin{equation*}
\mathbf{A}_2=\begin{pmatrix}
\omega + \omega x  + x^2 +x^3&  \omega x +\omega x^2 \\
 0 & 1
\end{pmatrix}.
\end{equation*}
\hfill $\diamond$
\end{example}

As shown in Example \ref{Code1-Code2-R1}, the Hermite normal form can be used to reduce a GPM in an upper triangular form. While this form can be achieved for any MT code, there exists the special case in which an MT code possesses a diagonal GPM. It is easily seen that a $\Lambda$-MT code $\mathcal{C}$ with block lengths $(m_1, m_2, \ldots, m_\ell)$ has a diagonal GPM if and only if $\mathcal{C}$ is the direct sum of $\ell$ constacyclic codes. Precisely, if $\mathcal{C}$ has a GPM of the form
$$\mathbf{G}=\mathrm{diag}\left(g_i(x) \right),$$
then $\mathcal{C}=\bigoplus_{i=1}^\ell \mathcal{C}_i$, where $\mathcal{C}_i$ is a $\lambda_i$-constacyclic code of length $m_i$ for $i=1,2,\ldots,\ell$. Other special forms of GPMs give rise to other subclasses of MT codes. Since our objective is to give a study that is valid for all MT codes, we do not impose any specific structure on the GPMs. However, as stated in Section \ref{Sec-Intro}, a particularly significant subclass is that of QC codes. A QC code is a special case of an MT code in which $m_1 = m_2 = \cdots = m_\ell = m$ and $\lambda_1 = \lambda_2 = \cdots = \lambda_\ell = 1$. Here, $m$ is referred to as the code co-index; consequently, a QC code has length $m\ell$. Since QC codes are MT codes, they can be represented by GPMs that satisfy identical equations. For instance, if $\mathcal{Q}$ is a QC code with index $\ell$, co-index $m$, and GPM $\mathbf{Q}$, then there exists an $\ell \times \ell$ polynomial matrix $\mathbf{P}$ satisfying the identical equation
\begin{equation}
\label{Identical_Eq_QC}
\mathbf{P}\mathbf{Q}=\mathrm{diag}\left(x^{m}- 1 \right)=(x^m-1)\mathbf{I}_\ell.
\end{equation}
It follows from \eqref{Dim_Eq_fromA} that $\dim{\mathcal{Q}}=\deg\left(\mathrm{det}\left( \mathbf{P}\right)\right)$, and from \eqref{Identical_Eq_QC} that $\mathbf{P}$ and $\mathbf{Q}$ commute.

\begin{example}
\label{Code3-Code4-R1}
Consider the $[9,3,6]$ code $\mathcal{C}_3$ with generator matrix $G_3$, as presented in Example \ref{Code3-Code4}. In fact, $\mathcal{C}_3$ is a $(1,1,1)$-MT code with index $\ell = 3$ and block lengths $(3,3,3)$. In other words, $\mathcal{C}_3$ is a QC code with co-index $m=3$. Following the construction outlined in Example \ref{Code1-Code2-R1}, we obtain a GPM $\mathbf{G}_3$ and a polynomial matrix $\mathbf{A}_3$ that satisfies the identical equation $\mathbf{A}_3 \mathbf{G}_3=(x^3-1)\mathbf{I}_3$ as follows
\begin{equation*}
\mathbf{G}_3=\begin{pmatrix}
1 & 1 + x +2 x^2 & 1+x \\
0 & x^3 -1 & 0\\
0 & 0 & x^3 -1
\end{pmatrix} \quad \text{and} \quad
\mathbf{A}_3=\begin{pmatrix}
x^3 -1 & 2 + 2x + x^2 & 2+2x \\
0 & 1 & 0\\
0 & 0 & 1
\end{pmatrix}.
\end{equation*}

Similarly, consider the $[9,7,1]$ code $\mathcal{C}_4$ with generator matrix $G_4$, as presented in Example \ref{Code3-Code4}. In fact, $\mathcal{C}_4$ is a $(2,1,2)$-MT code with index $\ell=3$ and block lengths $(3,3,3)$. We obtain the following GPM $\mathbf{G}_4$ and the polynomial matrix $\mathbf{A}_4$ satisfying its identical equation
\begin{equation*}
\mathbf{G}_4=\begin{pmatrix}
1+x&0&0\\
0&1&0\\
0&0&1+x
\end{pmatrix}\quad \text{and} \quad
\mathbf{A}_4=\begin{pmatrix}
1+2x+x^2&0&0\\
0&x^3-1&0\\
0&0&1+2x+x^2
\end{pmatrix}.
\end{equation*}
Since $\mathbf{G}_4$ is a diagonal polynomial matrix, it follows that $\mathcal{C}_4$ is the direct sum of three constacyclic codes, each of length $3$ with corresponding shift constants $2$, $1$, and $2$. 
\hfill $\diamond$
\end{example}

Let $\mathcal{C}$ be a $\Lambda$-MT code with block lengths $(m_1,m_2,\ldots,m_\ell)$ and a GPM $\mathbf{G}$, where $\Lambda=(\lambda_1,\lambda_2,\ldots,\lambda_\ell)$. It was shown in \cite[Theorems 5 and 8]{TakiEldin2023} that the Euclidean dual $\mathcal{C}^\perp$ (respectively, the $\kappa$-Galois dual $\mathcal{C}^{\perp_\kappa}$) is a $\Delta$-MT (respectively, $\sigma^{e-\kappa}\left(\Delta\right)$-MT) code with block lengths $(m_1,m_2,\ldots,m_\ell)$, where
\begin{align*}
\Delta&=\left(\lambda_1^{-1},\lambda_2^{-1},\ldots,\lambda_\ell^{-1}\right) =\Lambda^{-1}\\
\sigma^{e-\kappa}\left(\Delta\right)&= \left(\sigma^{e-\kappa}\left(\lambda_1^{-1}\right), \sigma^{e-\kappa}\left(\lambda_2^{-1}\right),\ldots, \sigma^{e-\kappa}\left(\lambda_\ell^{-1}\right)\right)=\sigma^{e-\kappa}\left(\Lambda^{-1}\right).
\end{align*}
In addition, $\mathcal{C}^\perp$ is the smallest submodule of $\left(\mathbb{F}_q[x]\right)^\ell$ that contains both the submodule $\bigoplus_{i=1}^\ell \langle x^{m_i}- \lambda_i^{-1}\rangle$ and the row space of
$$\mathbf{A}^T\!\left(\frac{1}{x}\right) \mathrm{diag}\left(x^{m_i}\right),$$
where $\mathbf{A}$ is the polynomial matrix satisfying the identical equation of $\mathbf{G}$. The notation $\mathbf{A}\!\left(\frac{1}{x}\right)$ refers to the matrix obtained from $\mathbf{A}$ by replacing $x$ with $\frac{1}{x}$, and this notation will be used frequently throughout the paper. Based on the construction outlined in Example \ref{Code1-Code2-R1}, a polynomial matrix that generates $\mathcal{C}^\perp$ can be explicitly formulated as
\begin{equation}
\label{Eq-Reciprocal1}
\begin{pmatrix}
\mathbf{A}^T\!\left(\frac{1}{x}\right) \mathrm{diag}\left(x^{m_i}\right)\\
\mathrm{diag}\left(x^{m_i}- \lambda_i^{-1} \right)
\end{pmatrix}.
\end{equation}
This polynomial matrix can be reduced to a GPM for $\mathcal{C}^\perp$ by performing elementary row operations over $\mathbb{F}_q[x]$. Throughout this paper, we denote a GPM for $\mathcal{C}^\perp$ by $\mathbf{H}$, and the polynomial matrix satisfying its identical equation by $\mathbf{B}$. Since $\mathcal{C}^\perp$ is a $\Delta$-MT code, the identical equation of $\mathbf{H}$ takes the form
\begin{equation}
\label{Identical_Eq_dual}
\mathbf{B}\mathbf{H}=\mathrm{diag}\left(x^{m_i}- \lambda_i^{-1} \right).
\end{equation}
From the discussion prior to Corollary \ref{Th-LinearGaloisIntersection}, it follows that $\sigma^{e-\kappa}\left(\mathbf{H}\right)$ is a GPM for the $\kappa$-Galois dual code $\mathcal{C}^{\perp_\kappa}$, where $\sigma^{e-\kappa}$ acts on the coefficients of the polynomial entries of $\mathbf{H}$. Moreover, an identical equation for $\mathcal{C}^{\perp_\kappa}$ is obtained by applying $\sigma^{e-\kappa}$ to both sides of \eqref{Identical_Eq_dual}; namely,
\begin{equation}
\label{Identical_Eq_dual_Galois}
\sigma^{e-\kappa}\left(\mathbf{B}\right) \sigma^{e-\kappa}\left(\mathbf{H}\right) =\mathrm{diag}\left(x^{m_i}- \sigma^{e-\kappa}\left(\lambda_i^{-1} \right) \right).
\end{equation}

\begin{example}
\label{Code1-Code2-R2}
The code $\mathcal{C}_1$ presented in Example \ref{Code1-Code2-R1} was shown to be $\Lambda$-MT, where $\Lambda=(1,\omega)$. Consequently, $\mathcal{C}_1^\perp$ is a $\Delta$-MT code with block lengths $(6,2)$, where $\Delta=(1,\omega^2)$. From \eqref{Eq-Reciprocal1}, a polynomial matrix generating $\mathcal{C}_1^\perp$ can be constructed as
\begin{equation*}
\begin{pmatrix}
\mathbf{A}_1^T\!\left(\frac{1}{x}\right) \mathrm{diag}\left(x^{m_i}\right)\\
\mathrm{diag}\left(x^{m_i}- \lambda_i^{-1} \right)
\end{pmatrix}=
\begin{pmatrix}
\omega^2 x^6+ \omega x^5 + x^4 + \omega^2 x^3 + \omega x^2 +x  &  0 \\
\omega x^6+ \omega x^5 + \omega x^3 +\omega x^2 & \omega^2 x^2+x \\
x^6-1 &0 \\
0 & x^2-\omega^2
\end{pmatrix}.
\end{equation*}
To obtain a GPM for $\mathcal{C}_1^\perp$, this polynomial matrix is reduced to its Hermite normal form, yielding the following reduced GPM $\mathbf{H}_1$ along with the polynomial matrix $\mathbf{B}_1$ that satisfies its identical equation
\begin{equation*}
\mathbf{H}_1=\begin{pmatrix}
 1+ x + x^3 +x^4   & \omega+x \\
0 &   \omega^2+x^2
\end{pmatrix}\quad \text{and}\quad
\mathbf{B}_1=\begin{pmatrix}
1+x+x^2  & \omega^2+x \\
0 &  1
\end{pmatrix}.
\end{equation*}
However, the $1$-Galois dual $\mathcal{C}_1^{\perp_1}$ of $\mathcal{C}$ is an MT code with the same block lengths, but with shift constants given by $\sigma^{e-\kappa}\left(\Delta\right)=\sigma\left(1,\omega^2\right)=\left(1,\omega\right)$. Moreover, since $\mathcal{C}_1^{\perp_1}=\sigma^{e-\kappa}\left(\mathcal{C}_1^\perp \right)=\sigma\left(\mathcal{C}_1^\perp\right)$, the corresponding GPM and the polynomial matrix satisfying its identical equation are given by
\begin{equation*}
\sigma\left(\mathbf{H}_1\right)=\begin{pmatrix}
 1+ x + x^3 +x^4   & \omega^2+x \\
0 &   \omega+x^2
\end{pmatrix}\quad \text{and}\quad
\sigma\left(\mathbf{B}_1\right)=\begin{pmatrix}
1+x+x^2  & \omega+x \\
0 &  1
\end{pmatrix}.
\end{equation*}
\hfill $\diamond$
\end{example}

So far, we have established that both the Euclidean and Galois duals of an MT code are MT. We have also shown how to construct their corresponding GPMs. In the remainder of this section, we prove that the reversed code of an MT code is also MT and provide a formula for its GPM. To this end, for a given block lengths $\left(m_1,m_2,\ldots,m_\ell\right)$, we define a permutation map $\mathcal{L}$ of $n$ elements, which acts on the vector $\mathbf{a}\in\mathbb{F}_q^n$, as given in \eqref{random_vector}, as follows:
$$\mathcal{L}\left(\mathbf{a}\right)=\left( a_{m_1-1,1},a_{m_1-2,1},\ldots,a_{1,1}, a_{0,1}, a_{m_2-1,2},a_{m_2-2,2},\ldots,a_{1,2}, a_{0,2},\ldots, a_{m_\ell-1,\ell},a_{m_\ell-2,\ell},\ldots,a_{1,\ell}, a_{0,\ell}\right).$$
Throughout the paper, $\mathcal{L}$ will be used frequently, and unless specifically mentioned, the block lengths will be deduced from the context. Given an MT code $\mathcal{C}$ with block lengths $\left(m_1,m_2,\ldots,m_\ell\right)$, we denote by $\mathcal{L}\left(\mathcal{C}\right)$ the code obtained by permuting each codeword of $\mathcal{C}$ by $\mathcal{L}$. Precisely,
$$\mathcal{L}\left(\mathcal{C}\right)=\left\{\mathcal{L}\left(\mathbf{c}\right) \quad \forall \mathbf{c}\in\mathcal{C}\right\}.$$
The following result shows that if $\mathcal{C}$ is an MT code, then $\mathcal{L}\left(\mathcal{C}\right)$ is also MT and provides its GPM. This result is fundamental in deriving a GPM for the reversed code of an MT code.

\begin{lemma}
\label{Th-Reciprocal}
Let $\Lambda=\left(\lambda_1,\lambda_2, \ldots, \lambda_\ell\right)$ and $\Delta=\left(\lambda_1^{-1}, \lambda_2^{-1}, \ldots, \lambda_\ell^{-1}\right)$. Consider a $\Lambda$-MT code $\mathcal{C}$ with block lengths $(m_1,m_2,\ldots,m_\ell)$. Then, the code $\mathcal{L}\left(\mathcal{C}\right)$ is $\Delta$-MT with block lengths $(m_1,m_2,\ldots,m_\ell)$ and a GPM $\mathbf{B}^T$, where $\mathbf{B}$ is the polynomial matrix satisfying \eqref{Identical_Eq_dual} for some GPM $\mathbf{H}$ of the dual code $\mathcal{C}^\perp$.
\end{lemma}
\begin{proof}
We first show that $\mathcal{L}\left(\mathcal{C}\right)$ is a $\Delta$-MT code. Consider an arbitrary codeword $\mathcal{L}\left(\mathbf{b}\right)\in \mathcal{L}\left(\mathcal{C}\right)$. Then, we have 
$$\mathbf{b}=\mathcal{L}\left(\mathcal{L}\left(\mathbf{b}\right)\right)\in \mathcal{L}\left(\mathcal{L}\left(\mathcal{C}\right)\right)=\mathcal{C}.$$
Thus, there exists $\mathbf{c}\in\mathcal{C}$ such that $\mathbf{b}=\mathcal{T}_\Lambda\left(\mathbf{c}\right)$. Consequently, $\mathcal{L}\left(\mathcal{C}\right)$ is $\Delta$-MT because
$$\mathcal{T}_\Delta\left(\mathcal{L}\left(\mathbf{b}\right)\right)=\mathcal{T}_\Delta\left(\mathcal{L}\left(\mathcal{T}_\Lambda\left(\mathbf{c}\right)\right)\right)=\mathcal{L}\left(\mathbf{c}\right)\in \mathcal{L}\left(\mathcal{C}\right).$$

Now, we derive a polynomial matrix that generates $\mathcal{L}\left(\mathcal{C}\right)$. In fact, for any polynomial vector $\mathbf{c}\left(x\right)=\left(c_1(x), c_2(x), \ldots, c_\ell(x)\right)\in\mathcal{C}$, the polynomial vector  
$$\mathbf{c}\left(\frac{1}{x}\right)\mathrm{diag}\left(x^{m_i}\right)=x\left(x^{m_1-1}c_1\left(\frac{1}{x}\right), x^{m_2-1}c_2\left(\frac{1}{x}\right), \ldots, x^{m_\ell-1}c_\ell \left(\frac{1}{x}\right)\right)\in \mathcal{L}\left(\mathcal{C}\right).$$
This shows that if $\mathbf{G}$ is a GPM for $\mathcal{C}$, then $\mathcal{L}\left(\mathcal{C}\right)$ is generated by the polynomial matrix 
\begin{equation}
\label{Inproof_Th-Reciprocal}
\begin{pmatrix}
\mathbf{G}\!\left(\frac{1}{x}\right) \mathrm{diag}\left(x^{m_i}\right)\\
\mathrm{diag}\left(x^{m_i}- \lambda_i^{-1} \right)
\end{pmatrix}.
\end{equation}

Now, we consider $\mathcal{L}\left(\mathcal{C}^\perp\right)$. Since $\mathcal{C}^\perp$ is a $\Delta$-MT code, it follows that $\mathcal{L}\left(\mathcal{C}^\perp\right)$ is a $\Lambda$-MT code. As was observed in the previous paragraph, a polynomial matrix generating $\mathcal{L}\left(\mathcal{C}^\perp\right)$ can be derived from a polynomial matrix generating $\mathcal{C}^\perp$. Using the polynomial matrix given by \eqref{Eq-Reciprocal1} for $\mathcal{C}^\perp$, we observe that $\mathbf{A}^T$ generates $\mathcal{L}\left(\mathcal{C}^\perp\right)$. More precisely, $\mathbf{A}^T$ is a GPM for $\mathcal{L}\left(\mathcal{C}^\perp\right)$, where $\mathbf{A}$ is the polynomial matrix satisfying the identical equation given by \eqref{Identical_Eq} for $\mathcal{C}$. However, our goal is to determine a GPM for $\mathcal{L}\left(\mathcal{C}\right)=\mathcal{L}\left(\left(\mathcal{C}^\perp\right)^\perp \right)$. Since $\mathbf{B}$ is the polynomial matrix satisfying the identical equation given by \eqref{Identical_Eq_dual} for $\mathcal{C}^\perp$, it follows that $\mathbf{B}^T$ is a GPM for $\mathcal{L}\left(\mathcal{C}\right)$.
\end{proof}

We are now prepared to present the main result of this section. Specifically, we show that the reversed code $\mathcal{R}$ of an MT code $\mathcal{C}$ is also an MT code. Additionally, we provide a GPM for $\mathcal{R}$.

\begin{theorem}
\label{Th-Reversed}
Let $\Lambda=\left(\lambda_1,\lambda_2, \ldots, \lambda_{\ell-1}, \lambda_\ell\right)$ and $\Gamma=\left(\lambda_\ell^{-1},\lambda_{\ell-1}^{-1}, \ldots, \lambda_2^{-1}, \lambda_1^{-1} \right)$. Consider a $\Lambda$-MT code $\mathcal{C}$ with block lengths $(m_1, m_2, \ldots, m_{\ell-1}, m_\ell)$. Then, the reversed code $\mathcal{R}$ of $\mathcal{C}$ is $\Gamma$-MT with block lengths $(m_\ell, m_{\ell-1},\ldots, m_2, m_1)$ and a GPM $\mathbf{B}^T  J_\ell$, where $\mathbf{B}$ is the polynomial matrix satisfying \eqref{Identical_Eq_dual} for some GPM $\mathbf{H}$ of the dual code $\mathcal{C}^\perp$ and $J_\ell$ is the $\ell\times \ell$ backward identity matrix.
\end{theorem}
\begin{proof}
We denote by $\mathcal{J}$ a permutation map of $n$ elements, which acts on the vector $\mathbf{a}\in\mathbb{F}_q^n$, as given in \eqref{random_vector}, as follows:
$$\mathcal{J}\left(\mathbf{a}\right)=\left( a_{0,\ell},a_{1,\ell},\ldots, a_{m_\ell-1,\ell}, \ldots, a_{0,2},a_{1,2},\ldots, a_{m_2-1,2}, a_{0,1},a_{1,1},\ldots, a_{m_1-1,1}\right).$$
In the polynomial vector representation, $\mathcal{J}$ acts on the polynomial vector $\mathbf{a}(x)$, as given in \eqref{random_poly_vector}, as follows:
$$\mathcal{J}\left(\mathbf{a}(x)\right)=\left(a_\ell(x), \ldots, a_2(x), a_1(x) \right).$$
This shows that if $\mathbf{G}$ is a GPM for $\mathcal{C}$, then the code $\mathcal{J}(\mathcal{C})$, obtained by applying $\mathcal{J}$ to each codeword of $\mathcal{C}$, is a $\left(\lambda_\ell,\lambda_{\ell-1}, \ldots, \lambda_2, \lambda_1 \right)$-MT code with block lengths $(m_\ell, m_{\ell-1}, \ldots, m_2, m_1)$ and a GPM given by $\mathbf{G} J_\ell$.

Our goal is to determine a GPM for $\mathcal{R}$. By combining Lemma \ref{Th-Reciprocal} with the observation in the previous paragraph, the result then follows from the fact that
$$\mathcal{R}=\mathcal{J}\left(\mathcal{L}\left(\mathcal{C}\right)\right).$$
\end{proof}

\begin{example}
\label{Code1-Code2-R6}
A direct application of Theorem \ref{Th-Reversed} ensures that the reversed code $\mathcal{R}_1$ of the MT code $\mathcal{C}_1$, presented in Example \ref{Code1-Code2-R1}, is MT. Since $\mathcal{C}_1$ is a $(1,\omega)$-MT code with block lengths $(6,2)$, it follows that $\mathcal{R}_1$ is a $(\omega^2,1)$-MT code with block lengths $(2,6)$. Moreover, a GPM for $\mathcal{R}_1$ can be easily derived from the polynomial matrix $\mathbf{B}_1$, given in Example \ref{Code1-Code2-R2}, as
\begin{equation*}
\mathbf{B}_1^T  J_2=
\begin{pmatrix}
0 & 1+x+x^2  \\
1 & \omega^2+x 
\end{pmatrix}.
\end{equation*}
By reducing this GPM to its Hermite normal form, the reduced GPM for $\mathcal{R}_1$ is
\begin{equation*}
\begin{pmatrix}
1 &  \omega^2+x \\
0 &   1+ x+x^2
\end{pmatrix}.
\end{equation*}
\hfill $\diamond$
\end{example}

\section{Intersection of MT codes}
\label{Sec-Intersection-MT-Codes}
In Section \ref{Sec-MT-Reversed-Code}, we showed how MT codes can be effectively represented using GPMs. In this section, we study the intersection of a pair of MT codes with identical block lengths. Our main objective is to employ GPMs to derive a formula for constructing a GPM for the intersection, provided that the intersection has an MT structure. This consideration is necessary because the intersection of two MT codes is not necessarily an MT code. In fact, the intersection of a pair of MT codes with different shift constants may or may not be MT. However, when the two MT codes have the same shift constants, their intersection is MT with the same shift constants. To summarize the objectives of this section, let $\mathcal{C}_1$ and $\mathcal{C}_2$ be a $\Lambda$-MT code and a $\Delta$-MT code, respectively, both with block lengths $(m_1,m_2,\ldots,m_\ell)$. We establish the following results:
\begin{enumerate}
\item We show that $\mathcal{C}_1\cap \mathcal{C}_2$ is not necessarily an MT code when $\Lambda\ne \Delta$. This is illustrated through Example \ref{Code3-Code4-R4}. While this result seems reasonable, we explicitly include it to confirm that the result established in \cite{LiuLiu2022} for constacyclic codes---particularly \cite[Example1]{LiuLiu2022}---remains valid in the broader class of MT codes.
\item We show that a pair of MT codes with different shift constants may still intersect in an MT code, as shown in Example \ref{Code3-Code4-R5}. In addition, when $\Lambda\ne \Delta$ and $\mathcal{C}_1 \cap \mathcal{C}_2$ admits an MT structure, we present Theorem \ref{Th-CondOfInter}, which proposes the shift constants for $\mathcal{C}_1\cap \mathcal{C}_2$. 
\item When $\Lambda= \Delta$, the intersection is a $\Lambda$-MT code, making it meaningful to define a GPM for the intersection. In this case, we prove Theorem \ref{Th-Intersection}, which provides a formula for constructing a GPM for $\mathcal{C}_1\cap \mathcal{C}_2$. 
\end{enumerate}

We begin with an example illustrating that the intersection of a pair of MT codes may result in a non-MT structure. In such cases, constructing a GPM for the intersection is meaningless. However, since the intersection remains a linear code, a generator matrix for this intersection can still be determined using Theorem \ref{Th-Linear}.

\begin{example}
\label{Code3-Code4-R4}
Consider the MT codes $\mathcal{C}_3$ and $\mathcal{C}_4$ over $\mathbb{F}_3$ with block lengths $(3,3,3)$, as presented in Example \ref{Code3-Code4-R1}. It was shown that $\mathcal{C}_3$ is a $(1,1,1)$-MT code, while $\mathcal{C}_4$ is a $(2,1,2)$-MT code. In addition, Theorem \ref{Th-Linear} was applied in Example \ref{Code3-Code4} to determine a generator matrix for their intersection $\mathcal{C}_3 \cap \mathcal{C}_4$, which is given by 
\begin{equation*}
\begin{pmatrix}
1&1&0&0&2&0&1&2&1
\end{pmatrix}.
\end{equation*}
Examining the code generated by this matrix, we observe that it does not offer an MT structure with block lengths $(3,3,3)$ for any choice of shift constants. This confirms that the intersection of a pair of MT codes is not necessarily MT.
\hfill $\diamond$
\end{example}

As shown in Example \ref{Code3-Code4-R4}, Theorem \ref{Th-Linear} can be used to determine the intersection of a pair of MT codes when the intersection does not offer an MT structure. Now, assuming that the intersection admits an MT structure, we aim to determine its shift constants when $\Lambda\ne\Delta$. The following result proposes shift constants for the intersection based on the minimum distances of the two codes. For completeness, we include the case when $\Lambda=\Delta$.

\begin{theorem}
\label{Th-CondOfInter}
Let $\mathcal{C}_1$ and $\mathcal{C}_2$ be $\Lambda$-MT and $\Delta$-MT codes over $\mathbb{F}_q$ with block lengths $(m_1, m_2, \ldots, m_\ell)$ and minimum distances $d\left(\mathcal{C}_1\right)$ and $d\left(\mathcal{C}_2\right)$, respectively. Define $D\left(\Lambda-\Delta\right)$ as the number of indices for which $\Lambda$ differs from $\Delta$. Then, 
\begin{enumerate}
\item If $\Lambda = \Delta$, then $\mathcal{C}_1\cap\mathcal{C}_2$ is a $\Lambda$-MT code with block lengths $(m_1,m_2,\ldots,m_\ell)$.
\item If $\Lambda \neq \Delta$ and $\mathcal{C}_1 \cap \mathcal{C}_2$ admits an MT structure with block lengths $(m_1,m_2,\ldots,m_\ell)$, then:
\begin{enumerate}
\item If $d\left(\mathcal{C}_1\right) > \ell$, then $\mathcal{C}_1\cap\mathcal{C}_2$ is a $\Lambda$-MT code.
\item If $d\left(\mathcal{C}_2\right) > \ell$, then $\mathcal{C}_1\cap\mathcal{C}_2$ is a $\Delta$-MT code.
\item If $d\left(\mathcal{C}_1\right) \leq \ell$ and $d\left(\mathcal{C}_2\right) \leq \ell$, then one of the subsequent possibilities occurs:
\begin{enumerate}
\item $\mathcal{C}_1\cap\mathcal{C}_2$ is simultaneously a $\Lambda$-MT and $\Delta$-MT code.
\item $\mathcal{C}_1\cap\mathcal{C}_2$ is a $\Lambda$-MT code, provided that $d\left(\mathcal{C}_2\right) \le D\left(\Lambda-\Delta\right)$. 
\item $\mathcal{C}_1\cap\mathcal{C}_2$ is a $\Delta$-MT code, provided that $d\left(\mathcal{C}_1\right) \le D\left(\Lambda-\Delta\right)$. 
\item $\mathcal{C}_1\cap\mathcal{C}_2$ is a $\Gamma$-MT code, provided that $d\left(\mathcal{C}_1\right) \le D\left(\Lambda-\Gamma\right)$ and $d\left(\mathcal{C}_2\right) \le D\left(\Delta-\Gamma\right)$. 
\end{enumerate}
\end{enumerate}
\end{enumerate}
\end{theorem}
\begin{proof}
If $\Lambda = \Delta$, then for every $\mathbf{c} \in \mathcal{C}_1 \cap \mathcal{C}_2$, we have $\mathcal{T}_\Lambda(\mathbf{c}) \in \mathcal{C}_1 \cap \mathcal{C}_2$. Consequently, $\mathcal{C}_1 \cap \mathcal{C}_2$ is a $\Lambda$-MT code.

Now, assume that $\Lambda \neq \Delta$ and that $\mathcal{C}_1 \cap \mathcal{C}_2$ admits an MT structure. If $\mathcal{T}_\Lambda\left(\mathbf{c}\right)\in\mathcal{C}_1\cap\mathcal{C}_2$ for every $\mathbf{c}\in\mathcal{C}_1\cap\mathcal{C}_2$, then $\mathcal{C}_1\cap\mathcal{C}_2$ is a $\Lambda$-MT code. Suppose, on the contrary, that $\mathcal{C}_1 \cap \mathcal{C}_2$ is not a $\Lambda$-MT code. Then there exists some $\mathbf{c} \in \mathcal{C}_1 \cap \mathcal{C}_2$ such that $\mathcal{T}_\Lambda\left(\mathbf{c}\right) \not\in \mathcal{C}_1 \cap \mathcal{C}_2$. By assumption, $\mathcal{C}_1 \cap \mathcal{C}_2$ is an MT code with shift constants, say, $\Gamma \neq \Lambda$. Therefore, we have $\mathcal{T}_\Gamma\left(\mathbf{c}\right) \in \mathcal{C}_1 \cap \mathcal{C}_2$. It follows that the difference $\mathcal{T}_\Gamma\left(\mathbf{c}\right)-\mathcal{T}_\Lambda\left(\mathbf{c}\right)$ is a nonzero codeword of $\mathcal{C}_1$. Moreover, the number of nonzero coordinates in this codeword is at most $D\left(\Lambda-\Gamma\right)$. Consequently, we obtain
$$d\left(\mathcal{C}_1\right)\le \text{weight}\left(\mathcal{T}_\Gamma\left(\mathbf{c}\right)-\mathcal{T}_\Lambda\left(\mathbf{c}\right)\right)\le D\left(\Lambda-\Gamma\right).$$

In summary, $\mathcal{C}_1 \cap \mathcal{C}_2$ is either a $\Lambda$-MT code or a $\Gamma$-MT code with the constraint $d\left(\mathcal{C}_1\right)\leq D\left(\Lambda-\Gamma\right)$. This latter case is valid only when $d\left(\mathcal{C}_1\right) \leq \ell$, implying that $\mathcal{C}_1 \cap \mathcal{C}_2$ must be a $\Lambda$-MT code if $d\left(\mathcal{C}_1\right)>\ell$. A similar conclusion holds upon replacing $\Lambda$ with $\Delta$ and $\mathcal{C}_1$ with $\mathcal{C}_2$.
\end{proof}

\begin{example}
\label{Code3-Code4-R5}
Consider the $[9,6,3]$ code $\mathcal{C}_5$ with generator matrix $G_5$, as presented in Example \ref{Code3-Code4}. In fact, $\mathcal{C}_5$ is a $\Delta$-MT code with index $\ell = 3$ and block lengths $(3,3,3)$, where $\Delta=(2,2,2)$. In other words, $\mathcal{C}_5$ is a $\lambda$-QT code with $\lambda=2$. Following the construction outlined in Example \ref{Code1-Code2-R1}, we obtain its reduced GPM $\mathbf{G}_5$ and the polynomial matrix $\mathbf{A}_5$ that satisfies the identical equation as follows
\begin{equation*}
\mathbf{G}_5=\begin{pmatrix}
1&0&1+2 x^2\\
0&1&1+x^2\\
0&0&x^3 - 2
\end{pmatrix}\quad \text{and} \quad
\mathbf{A}_5=\begin{pmatrix}
x^3 - 2 &0 & 2+x^2 \\
0 & x^3 - 2 &2+2 x^2 \\
0 & 0 & 1
\end{pmatrix} .
\end{equation*}

In this example, We examine the intersection $\mathcal{C}_3 \cap \mathcal{C}_5$, where $\mathcal{C}_3$ is the $\Lambda$-MT code with parameters $[9,3,6]$ and block lengths $(3,3,3)$, as presented in Example \ref{Code3-Code4-R1}, with shift constants $\Lambda = (1,1,1)$. In other words, $\mathcal{C}_3$ is a QC code. According to Theorem \ref{Th-CondOfInter}, since $d\left(\mathcal{C}_3\right)=6 > 3=\ell$, if $\mathcal{C}_3 \cap \mathcal{C}_5$ admits an MT structure, then it must be a $\Lambda$-MT code, i.e., QC.

Applying Theorem \ref{Th-Linear}, a generator matrix for $\mathcal{C}_3 \cap \mathcal{C}_5$ is determined as
\begin{equation*}
\begin{pmatrix}
1&1&1&1&1&1&2&2&2
\end{pmatrix}.
\end{equation*}
Thus, $\mathcal{C}_3 \cap \mathcal{C}_5$ is a maximum distance separable (MDS) code with parameters $[9,1,9]$. In addition, verifying its MT structure with block lengths $(3,3,3)$ and the shift constants $\Lambda=(1,1,1)$ proposed by Theorem \ref{Th-CondOfInter}, we confirm that $\mathcal{C}_3 \cap \mathcal{C}_5$ is indeed a $\Lambda$-MT code, i.e., QC. 
\hfill $\diamond$
\end{example}

In practical applications, codes with large minimum distances are of significant importance. Consequently, studying the intersection of MT codes with minimum distances greater than the code index $\ell$ is particularly relevant. The following consequence of Theorem \ref{Th-CondOfInter} provides an important insight into the intersection of MT codes with minimum distances exceeding the code index.

\begin{corollary}
\label{Corr-Th-CondOfInter}
Let $\Lambda=\left(\lambda_1,\lambda_2,\ldots,\lambda_\ell\right)$ and $\Delta=\left(\delta_1,\delta_2,\ldots,\delta_\ell\right)$. Suppose $\mathcal{C}_1$ and $\mathcal{C}_2$ are $\Lambda$-MT and $\Delta$-MT codes, respectively, with block lengths $(m_1, m_2, \ldots, m_\ell)$ and minimum distances satisfying $d\left(\mathcal{C}_1\right)>\ell$ and $d\left(\mathcal{C}_2\right)>\ell$. Then, $\mathcal{C}_1\cap\mathcal{C}_2$ admits an MT structure with block lengths $(m_1, m_2, \ldots, m_\ell)$ if and only if the projection of $\mathcal{C}_1\cap\mathcal{C}_2$ onto its $i$-th block is zero for every index $1\le i\le \ell$ such that $\lambda_i\ne\delta_i$. In particular, if $\lambda_i\ne\delta_i$ for all $i=1,2,\ldots,\ell$, then $\mathcal{C}_1\cap\mathcal{C}_2$ admits an MT structure with block lengths $(m_1, m_2, \ldots, m_\ell)$ if and only if $\mathcal{C}_1\cap\mathcal{C}_2=\left\{\mathbf{0}\right\}$.
\end{corollary}
\begin{proof}
Assume that the projection of $\mathcal{C}_1\cap\mathcal{C}_2$ onto its $i$-th block is zero for every $1\le i\le \ell$ with $\lambda_i\ne\delta_i$. Then, for any $\mathbf{c}\in \mathcal{C}_1\cap\mathcal{C}_2$, we have $\mathcal{T}_\Lambda\left(\mathbf{c}\right)=\mathcal{T}_\Delta\left(\mathbf{c}\right)$. Consequently, $\mathcal{C}_1\cap\mathcal{C}_2$ is simultaneously $\Lambda$-MT and $\Delta$-MT, implying that $\mathcal{C}_1\cap\mathcal{C}_2$ admits an MT structure. 

Conversely, assume that $\mathcal{C}_1\cap\mathcal{C}_2$ admits an MT structure. By Theorem \ref{Th-CondOfInter}, since $d\left(\mathcal{C}_1\right)>\ell$ and $d\left(\mathcal{C}_2\right)>\ell$, it follows that $\mathcal{C}_1\cap\mathcal{C}_2$ is simultaneously $\Lambda$-MT and $\Delta$-MT. Let $1\le i\le \ell$ be an index such that $\lambda_i\ne\delta_i$. Assume there exists a codeword $\mathbf{c} \in \mathcal{C}_1\cap\mathcal{C}_2$ with a nonzero $i$-th block. We may assume without loss of generality that $c_{m_i-1,i}\ne 0$, see \eqref{random_vector}. Then the difference $\mathcal{T}_\Lambda\left(\mathbf{c}\right)-\mathcal{T}_\Delta\left(\mathbf{c}\right)$ is nonzero. Since $\mathcal{C}_1\cap\mathcal{C}_2$ is simultaneously $\Lambda$-MT and $\Delta$-MT, $\mathcal{T}_\Lambda\left(\mathbf{c}\right)-\mathcal{T}_\Delta\left(\mathbf{c}\right)$ is a codeword of $\mathcal{C}_1\cap\mathcal{C}_2$, and moreover, it has weight at most $\ell$. But this contradicts the fact that $d\left(\mathcal{C}_1 \cap\mathcal{C}_2\right) \ge d\left(\mathcal{C}_1\right)>\ell$. Thus, every codeword $\mathbf{c} \in \mathcal{C}_1\cap\mathcal{C}_2$ must have a zero $i$-th block for every index $1\le i\le \ell$ such that $\lambda_i\ne\delta_i$.

In the special case where $\lambda_i\ne\delta_i$ for all $i=1,2,\ldots,\ell$, it follows that $\mathcal{C}_1\cap\mathcal{C}_2$ admits an MT structure if and only if every codeword $\mathbf{c} \in \mathcal{C}_1\cap\mathcal{C}_2$ has all of its blocks equal to zero, i.e., $\mathcal{C}_1\cap\mathcal{C}_2=\left\{\mathbf{0}\right\}$.
\end{proof}

\begin{remark}
\label{Imp-Remark}
Corollary \ref{Corr-Th-CondOfInter} shows that if $\Lambda \ne\Delta$, then the intersection of the corresponding MT codes will have trivial blocks at the indices where $\Lambda$ and $\Delta$ differ. We remark that, despite that the two codes ($\mathcal{C}_3$ and $\mathcal{C}_5$) presented in Example \ref{Code3-Code4-R5} satisfy the condition $\lambda_i\ne\delta_i$ for all $i=1,2,\ldots,\ell$, they have a nontrivial intersection. This does not contradict Corollary \ref{Corr-Th-CondOfInter} because one of the two codes, namely $\mathcal{C}_5$, has a minimum distance equal to the code index. In contrast, $\mathcal{C}_5^\perp$ is a $(2,2,2)$-MT code with block lengths $(3,3,3)$ and a minimum distance $d\left(\mathcal{C}_5^\perp\right) = 5 > \ell$. We find that $\mathcal{C}_3 \cap \mathcal{C}_5^\perp = \{\mathbf{0}\}$ by Corollary \ref{Linear_LCP}, since $\mathrm{rank}\left(G_5 G_3^T\right) = k_3$; this is consistent with Corollary \ref{Corr-Th-CondOfInter}. In fact, Corollary \ref{Corr-Th-CondOfInter} holds under the condition $d\left(\mathcal{C}_1\right)>\ell$ and $d\left(\mathcal{C}_2\right)>\ell$. This condition may be of primary interest in practical applications, since codes with small minimum distances are generally less relevant. On the other hand, an intersection with trivial blocks at the indices where $\Lambda$ and $\Delta$ differ may seem meaningless. Thus, Corollary \ref{Corr-Th-CondOfInter} naturally suggests that the case where $\lambda_i=\delta_i$ for all $1\le i\le \ell$, i.e., $\Lambda=\Delta$, seems the more relevant framework. Accordingly, hereinafter, we assume identical shift constants for the two codes whose intersection is under consideration. This assumption not only to prevent trivial blocks but also guarantees that the intersection always admits an MT structure, regardless of the minimum distances of the pair of codes. This is because, by setting $\Lambda=\Delta$, Theorem \ref{Th-CondOfInter} shows that the intersection is $\Lambda$-MT without imposing any constraints on the minimum distances.
\hfill $\diamond$
\end{remark}

Our next result, which is the main result of this section, provides a GPM construction for the intersection of a pair of $\Lambda$-MT codes. We introduce the following convenient notation that will be used in all subsequent results. Define $N$ as the smallest integer for which the map $\mathcal{T}_\Lambda^N$ acts as the identity map on $\mathbb{F}_q^n$. For a QC code, $N$ is the code co-index $m$. In contrast, for a $\lambda$-QT code, $N$ is given by $N=t m$, where $t$ is the multiplicative order of $\lambda$  in $\mathbb{F}_q$. More generally, for a $\left(\lambda_1,\lambda_2, \ldots, \lambda_\ell\right)$-MT code with block lengths $(m_1,m_2,\ldots,m_\ell)$, $N$ is the least common multiple of the integers $t_1 m_1, t_2 m_2, \ldots, t_\ell m_\ell$, where $t_i$ is the multiplicative order of $\lambda_i$. We denote $N$ by $\mathrm{lcm}(t_i m_i)$, and this notation will be used frequently in the remainder of the paper. In the proof of the next theorem, we will use the following identity. Let $\mathbf{G}$ be a GPM for a $\Lambda$-MT code with block lengths $(m_1,m_2,\ldots,m_\ell)$, and let $\mathbf{A}$ be the polynomial matrix satisfying \eqref{Identical_Eq}. Since $x^{m_i}-\lambda_i$ divides $x^N-1$ for each $i=1,2,\ldots,\ell$, we define the diagonal polynomial matrix 
$$\mathrm{diag}\left(\frac{x^N-1}{x^{m_i}-\lambda_i}\right),$$
where each diagonal entry is given by $(x^N-1)/(x^{m_i}-\lambda_i)$. Now, we observe that
$$\mathrm{diag}\left(\frac{x^N-1}{x^{m_i}-\lambda_i}\right) \mathbf{G}^T \mathbf{A}^T=\mathrm{diag}\left(\frac{x^N-1}{x^{m_i}-\lambda_i}\right)\left(\mathbf{A}\mathbf{G}\right)^T=\mathrm{diag}\left(\frac{x^N-1}{x^{m_i}-\lambda_i}\right) \mathrm{diag}\left(x^{m_i}- \lambda_i \right)=(x^N-1)\mathbf{I}_\ell.$$
This equation implies that the terms on the left-hand side can be interchanged, leading to the following identity:
\begin{equation}
\label{Eq-for-proof}
 \mathbf{A}^T \mathrm{diag}\left(\frac{x^N-1}{x^{m_i}-\lambda_i}\right) \mathbf{G}^T=(x^N-1)\mathbf{I}_\ell.
\end{equation}
This identity will be used in the proof of the following theorem.

\begin{theorem}
\label{Th-Intersection}
Let $\mathcal{C}_1$ and $\mathcal{C}_2$ be $\Lambda$-MT codes over $\mathbb{F}_q$ with index $\ell$, block lengths $(m_1,m_2,\ldots,m_\ell)$, and GPMs $\mathbf{G}_1$ and $\mathbf{G}_2$, respectively. Let $\mathbf{A}_1$ and $\mathbf{A}_2$ denote the corresponding polynomial matrices that satisfy the identical equation given by \eqref{Identical_Eq} for $\mathbf{G}_1$ and $\mathbf{G}_2$. Define $\mathcal{Q}$ as the QC code over $\mathbb{F}_q$ with index $\ell$ and co-index $N$, generated by the polynomial matrix
\begin{equation*}\begin{pmatrix}
\mathbf{A}_1^T \mathrm{diag}\left(\frac{x^N-1}{x^{m_i}-\lambda_i}\right) \mathbf{G}_2^T\\
(x^N-1)\mathbf{I}_\ell
\end{pmatrix},
\end{equation*}
where $N=\mathrm{lcm}(t_i m_i)$. Let $\mathbf{Q}$ be a GPM for $\mathcal{Q}$, and let $\mathbf{P}$ be the polynomial matrix satisfying the identical equation \eqref{Identical_Eq_QC} for $\mathbf{Q}$, that is, 
$$\mathbf{P}\mathbf{Q}=(x^N-1)\mathbf{I}_\ell.$$
Then, $\mathbf{P}^T\mathbf{G}_2$ is a GPM for the intersection $\mathcal{C}_1\cap \mathcal{C}_2$.
\end{theorem}
\begin{proof}
Suppose that $\overline{\mathbf{G}}$ is a GPM for $\mathcal{C}_1\cap \mathcal{C}_2$. Then, there exist polynomial matrices $\mathbf{R}_1$ and $\mathbf{R}_2$ such that $\overline{\mathbf{G}}=\mathbf{R}_1\mathbf{G}_1=\mathbf{R}_2\mathbf{G}_2$. In addition, there exists a polynomial matrix $\overline{\mathbf{A}}$ satisfying the identical equation 
$$\overline{\mathbf{A}}\overline{\mathbf{G}}=\mathrm{diag}(x^{m_i}-\lambda_i).$$
Then, for $i=1,2$, we have
$$\left(\overline{\mathbf{A}}\mathbf{R}_i\right)\mathbf{G}_i=\overline{\mathbf{A}}\left(\mathbf{R}_i\mathbf{G}_i\right)=\overline{\mathbf{A}}\overline{\mathbf{G}}=\mathrm{diag}(x^{m_i}-\lambda_i)=\mathbf{A}_i \mathbf{G}_i.$$
This implies that $\mathbf{A}_i=\overline{\mathbf{A}}\mathbf{R}_i$ due to the linear independence of the rows of $\mathbf{G}_i$.

We claim that 
\begin{equation}
\label{Th-Intersection-InProof}
\overline{\mathbf{A}}=\mathbf{A}_1\mathbf{M}_1+\mathbf{A}_2\mathbf{M}_2
\end{equation} 
for some polynomial matrices $\mathbf{M}_1$ and $\mathbf{M}_2$. To show this, we construct an $\ell\times \ell$ polynomial matrix $\mathbf{R}$ whose columns form a basis for the submodule of $\left(\mathbb{F}_q[x]\right)^\ell$ generated by the columns of $\mathbf{R}_1$ and $\mathbf{R}_2$. Consequently, there exist polynomial matrices $\mathbf{M}_1'$ and $\mathbf{M}_2'$ such that $\mathbf{R}=\mathbf{R}_1\mathbf{M}_1'+\mathbf{R}_2\mathbf{M}_2'$. Furthermore, $\mathbf{R}_1$ and $\mathbf{R}_2$ can be expressed as 
$$\mathbf{R}_1=\mathbf{R}\mathbf{X}_1 \quad \text{and}\quad \mathbf{R}_2=\mathbf{R}\mathbf{X}_2,$$ 
for some polynomial matrices $\mathbf{X}_1$ and $\mathbf{X}_2$ with nonzero determinants. Using this result, we get
$$\mathbf{R}\mathbf{X}_1\mathbf{G}_1=\mathbf{R}_1\mathbf{G}_1=\overline{\mathbf{G}}=\mathbf{R}_2\mathbf{G}_2=\mathbf{R}\mathbf{X}_2\mathbf{G}_2.$$
It follows that $\mathbf{R}\left(\mathbf{X}_1\mathbf{G}_1-\mathbf{X}_2\mathbf{G}_2\right)=\mathbf{0}$. Since the columns of $\mathbf{R}$ are linearly independent, we conclude that $\mathbf{X}_1\mathbf{G}_1=\mathbf{X}_2\mathbf{G}_2$. Consequently, the rows of $\mathbf{X}_1\mathbf{G}_1$ correspond to codewords of $\mathcal{C}_1\cap\mathcal{C}_2$. This implies the existence of a polynomial matrix $\mathbf{Y}$ such that $\mathbf{Y}\overline{\mathbf{G}}=\mathbf{X}_1\mathbf{G}_1$, which further yields $\mathbf{Y}\mathbf{R}_1=\mathbf{X}_1$ after replacing $\overline{\mathbf{G}}$ with $\mathbf{R}_1\mathbf{G}_1$. Since we have established that $\mathbf{Y}\mathbf{R}\mathbf{X}_1=\mathbf{Y}\mathbf{R}_1=\mathbf{X}_1$, and $\mathbf{X}_1$ has a nonzero determinant, it follows that $\mathbf{Y}\mathbf{R}=\mathbf{I}_\ell$. Thus, we obtain
$$\overline{\mathbf{A}}= \overline{\mathbf{A}}\mathbf{I}_\ell= \overline{\mathbf{A}}\mathbf{R}\mathbf{Y}= \overline{\mathbf{A}}\left(\mathbf{R}_1\mathbf{M}_1'+\mathbf{R}_2\mathbf{M}_2'\right)\mathbf{Y}= \overline{\mathbf{A}}\mathbf{R}_1\mathbf{M}_1' \mathbf{Y}+\overline{\mathbf{A}}\mathbf{R}_2\mathbf{M}_2' \mathbf{Y}=\mathbf{A}_1\mathbf{M}_1' \mathbf{Y}+\mathbf{A}_2\mathbf{M}_2' \mathbf{Y}.$$
This proves the claim \eqref{Th-Intersection-InProof} with $\mathbf{M}_1=\mathbf{M}_1' \mathbf{Y}$ and $\mathbf{M}_2=\mathbf{M}_2' \mathbf{Y}$.

Our next claim is that $\overline{\mathbf{A}}^T \mathrm{diag}\left(\frac{x^N-1}{x^{m_i}-\lambda_i}\right) \mathbf{G}_2^T$ forms a GPM for the QC code $\mathcal{Q}$. This claim is justified by the following two observations. Using \eqref{Eq-for-proof}, we observe that
\begin{equation*}
\begin{split}
\begin{pmatrix}
\mathbf{M}_1^T  & \mathbf{M}_2^T  
\end{pmatrix}\begin{pmatrix}
\mathbf{A}_1^T \mathrm{diag}\left(\frac{x^N-1}{x^{m_i}-\lambda_i}\right) \mathbf{G}_2^T\\
(x^N-1)\mathbf{I}_\ell
\end{pmatrix}&= \mathbf{M}_1^T \mathbf{A}_1^T \mathrm{diag}\left(\frac{x^N-1}{x^{m_i}-\lambda_i}\right) \mathbf{G}_2^T + \mathbf{M}_2^T (x^N-1)\\& =\mathbf{M}_1^T \mathbf{A}_1^T \mathrm{diag}\left(\frac{x^N-1}{x^{m_i}-\lambda_i}\right) \mathbf{G}_2^T + \mathbf{M}_2^T \mathbf{A}_2^T \mathrm{diag}\left(\frac{x^N-1}{x^{m_i}-\lambda_i}\right) \mathbf{G}_2^T \\& =\left( \mathbf{M}_1^T \mathbf{A}_1^T + \mathbf{M}_2^T \mathbf{A}_2^T \right)\mathrm{diag}\left(\frac{x^N-1}{x^{m_i}-\lambda_i}\right) \mathbf{G}_2^T \\&=\overline{\mathbf{A}}^T \mathrm{diag}\left(\frac{x^N-1}{x^{m_i}-\lambda_i}\right) \mathbf{G}_2^T 
\end{split}
\end{equation*}
and
\begin{equation*}
\begin{split}
\begin{pmatrix}
\mathbf{R}_1^T \\
\mathbf{R}_2^T
\end{pmatrix}\overline{\mathbf{A}}^T \mathrm{diag}\left(\frac{x^N-1}{x^{m_i}-\lambda_i}\right) \mathbf{G}_2^T &= \begin{pmatrix}
\mathbf{R}_1^T \overline{\mathbf{A}}^T \mathrm{diag}\left(\frac{x^N-1}{x^{m_i}-\lambda_i}\right) \mathbf{G}_2^T\\
\mathbf{R}_2^T \overline{\mathbf{A}}^T \mathrm{diag}\left(\frac{x^N-1}{x^{m_i}-\lambda_i}\right) \mathbf{G}_2^T
\end{pmatrix}\\& =
\begin{pmatrix}
\mathbf{A}_1^T \mathrm{diag}\left(\frac{x^N-1}{x^{m_i}-\lambda_i}\right) \mathbf{G}_2^T\\
\mathbf{A}_2^T \mathrm{diag}\left(\frac{x^N-1}{x^{m_i}-\lambda_i}\right) \mathbf{G}_2^T
\end{pmatrix} =
\begin{pmatrix}
\mathbf{A}_1^T \mathrm{diag}\left(\frac{x^N-1}{x^{m_i}-\lambda_i}\right) \mathbf{G}_2^T\\
(x^N-1)\mathbf{I}_\ell
\end{pmatrix}. 
\end{split}
\end{equation*}
Therefore, $\overline{\mathbf{A}}^T \mathrm{diag}\left(\frac{x^N-1}{x^{m_i}-\lambda_i}\right) \mathbf{G}_2^T$ forms a GPM for $\mathcal{Q}$ because it generates the polynomial matrix that generates $\mathcal{Q}$ and vice versa. Since $\mathbf{Q}$ is also a GPM for $\mathcal{Q}$, there exists an invertible matrix $\mathbf{U}$ such that $\mathbf{Q}=\mathbf{U}\overline{\mathbf{A}}^T \mathrm{diag}\left(\frac{x^N-1}{x^{m_i}-\lambda_i}\right) \mathbf{G}_2^T$.

Now, since
$$\mathbf{G}_2\mathrm{diag}\left(\frac{x^N-1}{x^{m_i}-\lambda_i}\right)\overline{\mathbf{A}}\mathbf{U}^T\mathbf{P}^T = \mathbf{Q}^T\mathbf{P}^T =(x^N-1)\mathbf{I}_\ell,$$
it follows that
$$\overline{\mathbf{A}}\mathbf{U}^T\mathbf{P}^T \mathbf{G}_2\mathrm{diag}\left(\frac{x^N-1}{x^{m_i}-\lambda_i}\right) = (x^N-1)\mathbf{I}_\ell = \mathrm{diag}\left( x^{m_i}-\lambda_i \right) \mathrm{diag}\left(\frac{x^N-1}{x^{m_i}-\lambda_i}\right) = \overline{\mathbf{A}}\overline{\mathbf{G}} \mathrm{diag}\left(\frac{x^N-1}{x^{m_i}-\lambda_i}\right).$$
Since the columns of $\overline{\mathbf{A}}$ are linearly independent, we get $\overline{\mathbf{G}}=\mathbf{U}^T \mathbf{P}^T \mathbf{G}_2$. Moreover, since $\mathbf{U}$ is invertible and $\overline{\mathbf{G}}$ is a GPM for $\mathcal{C}_1\cap \mathcal{C}_2$, we conclude that $\mathbf{P}^T \mathbf{G}_2$ is another GPM for $\mathcal{C}_1\cap \mathcal{C}_2$.
\end{proof}

\begin{example}
\label{Code1-Code2-R3}
The codes $\mathcal{C}_1$ and $\mathcal{C}_2$ presented in Example~\ref{Code1-Code2} were considered as linear codes and Theorem~\ref{Th-Linear} was employed to determine their intersection. In Example~\ref{Code1-Code2-R1}, we showed that $\mathcal{C}_1$ and $\mathcal{C}_2$ are $(1,\omega)$-MT codes with block lengths $(6,2)$. Hence, Theorem~\ref{Th-Intersection} enables computing the intersection $\mathcal{C}_1 \cap \mathcal{C}_2$ within the structure of $(1,\omega)$-MT codes, and further provides a GPM for this intersection. In the notation of Theorem~\ref{Th-Intersection}, we observe that $t_1 = 1$ and $t_2 = 3$, which implies that $N = \mathrm{lcm}(6,6) = 6$. Define $\mathcal{Q}$ as the QC code over $\mathbb{F}_4$ of index $2$ and co-index $6$, generated by the polynomial matrix 
\begin{equation*}
\begin{pmatrix}
\mathbf{A}_1^T \mathrm{diag}\left(\frac{x^6-1}{x^{m_i}-\lambda_i}\right) \mathbf{G}_2^T\\
(x^6-1)\mathbf{I}_2
\end{pmatrix}=
\begin{pmatrix}
\omega + \omega^2 x + x^2 + \omega x^6 + \omega^2 x^7 +x^8 & 0\\
 1 + \omega^2 x + \omega x^2 + \omega^2 x^3 + x^4 + \omega x^5  + \omega x^6 +\omega x^7 & \omega^2+ x + \omega^2 x^6 +x^7 \\
x^6-1 & 0\\
0 & x^6-1
\end{pmatrix}.
\end{equation*}
By reducing to Hermite normal form, we obtain the reduced GPM for $\mathcal{Q}$, as well as the corresponding polynomial matrix that satisfies its identical equation. These are, respectively, given by
\begin{equation*}
\mathbf{Q}=
\begin{pmatrix}
 \omega + \omega^2 x + x^2 + \omega x^3 + \omega^2 x^4 +x^5 &0\\
0&x^6 + 1
\end{pmatrix}\quad \text{and}\quad
\mathbf{P}=
\begin{pmatrix}
\omega^2+x &0 \\
0& 1
\end{pmatrix}.
\end{equation*}

Applying Theorem~\ref{Th-Intersection}, the following GPM is obtained for $\mathcal{C}_1 \cap \mathcal{C}_2$:
\begin{equation*}
\mathbf{P}^T\mathbf{G}_2=
\begin{pmatrix}
\omega+ x + \omega x^3 +x^4 & x+\omega x^2 \\
0 & \omega+x^2
\end{pmatrix}.
\end{equation*}
Hence, the reduced GPM for $\mathcal{C}_1 \cap \mathcal{C}_2$ is 
\begin{equation*}
\begin{pmatrix}
\omega + x + \omega x^3 +x^4 &  \omega^2 + x\\
0    &   \omega+x^2
\end{pmatrix}.
\end{equation*}
It can be verified that this polynomial matrix is the reduced GPM for the generator matrix obtained for $\mathcal{C}_1 \cap \mathcal{C}_2$ in Example~\ref{Code1-Code2} via Theorem~\ref{Th-Linear}.
\hfill $\diamond$
\end{example}

Theorem~\ref{Th-Intersection} for MT codes is the analogue of Theorem~\ref{Th-Linear} which deals with linear codes. Therefore, it is natural to present the analogue of Corollary~\ref{Th-LinearGaloisIntersection} in the context of MT codes. This leads to the following result, which is proven as a consequence of Theorem~\ref{Th-Intersection}. Specifically, we construct a GPM for the intersection of the $\kappa$-Galois dual of a $\Lambda$-MT code $\mathcal{C}_1$ with another $\Delta$-MT code $\mathcal{C}_2$. According to Remark~\ref{Imp-Remark}, we assume that $\mathcal{C}_1^{\perp_\kappa}$ and $\mathcal{C}_2$ have the same shift constants. Equivalently, since $\mathcal{C}_1^{\perp_\kappa}$ is a $\sigma^{e-\kappa}(\Lambda^{-1})$-MT code, we impose the condition $\sigma^{e-\kappa}(\Lambda^{-1}) = \Delta$. This is to say, if $\Lambda=\left(\lambda_1,\lambda_2,\ldots,\lambda_\ell\right)$ and $\Delta=\left(\delta_1,\delta_2,\ldots,\delta_\ell\right)$, we require that $\sigma^{e-\kappa}(\lambda_i^{-1}) = \delta_i$ for all $1 \le i \le \ell$.

\begin{corollary}
\label{Th-Galois-Intersection}
Let $\mathcal{C}_1$ and $\mathcal{C}_2$ be a $\Lambda$-MT code and a $\Delta$-MT code over $\mathbb{F}_{p^e}$, respectively, each with index $\ell$, block lengths $(m_1,m_2,\ldots,m_\ell)$, and GPMs $\mathbf{G}_1$ and $\mathbf{G}_2$. Let $\mathbf{A}_1$ and $\mathbf{A}_2$ denote the corresponding polynomial matrices satisfying the identical equation \eqref{Identical_Eq} for $\mathbf{G}_1$ and $\mathbf{G}_2$, respectively. Suppose that $\sigma^{e-\kappa}\left(\Lambda^{-1}\right)=\Delta$ for some $0\le \kappa <e$. Define $\mathcal{Q}$ as the QC code over $\mathbb{F}_{p^e}$ with index $\ell$ and co-index $N$, generated by the polynomial matrix
\begin{equation*}
\begin{pmatrix}
\sigma^{e-\kappa}\left(\mathbf{G}_1\left(\frac{1}{x}\right) \mathrm{diag}\left(x^{m_i}\right) \right)\mathrm{diag}\left(\frac{x^N-1}{x^{m_i}-\delta_i}\right) \mathbf{G}_2^T\\
(x^N-1)\mathbf{I}_\ell
\end{pmatrix},
\end{equation*}
where $N=\mathrm{lcm}(t_i m_i)$. Let $\mathbf{Q}$ be a GPM for $\mathcal{Q}$, and let $\mathbf{P}$ be the corresponding polynomial matrix satisfying the identical equation \eqref{Identical_Eq_QC} for $\mathbf{Q}$. Then, the intersection $\mathcal{C}^{\perp_\kappa}_1\cap \mathcal{C}_2$ is a $\Delta$-MT code with GPM $\mathbf{P}^T\mathbf{G}_2$.
\end{corollary}
\begin{proof}
From~\eqref{Identical_Eq_dual_Galois}, it follows that $\mathcal{C}_1^{\perp_\kappa}$ is a $\sigma^{e-\kappa}\left(\Lambda^{-1}\right)$-MT with GPM $\sigma^{e-\kappa}\left(\mathbf{H}_1\right)$; that is, it is a $\Delta$-MT code. Consequently, by Theorem~\ref{Th-CondOfInter}, the intersection $\mathcal{C}^{\perp_\kappa}_1\cap \mathcal{C}_2$ is also a $\Delta$-MT code. To determine a GPM for this intersection using Theorem~\ref{Th-Intersection}, we construct the associated QC code $\mathcal{Q}$, generated by
\begin{equation}
\label{Th-Galois-Intersection-inproof1}
\begin{pmatrix}
\sigma^{e-\kappa}\left(\mathbf{B}_1^T\right) \mathrm{diag}\left(\frac{x^N-1}{x^{m_i}-\delta_i}\right) \mathbf{G}_2^T\\
(x^N-1)\mathbf{I}_\ell
\end{pmatrix}.
\end{equation}
According to Lemma~\ref{Th-Reciprocal}, $\mathbf{B}_1^T$ is a GPM for $\mathcal{L}\left(\mathcal{C}_1\right)$. Furthermore, \eqref{Inproof_Th-Reciprocal} provides an alternative polynomial matrix generating $\mathcal{L}\left(\mathcal{C}_1\right)$ given by 
\begin{equation}
\label{Th-Galois-Intersection-inproof2}
\begin{pmatrix}
\mathbf{G}_1\left(\frac{1}{x}\right) \mathrm{diag}\left(x^{m_i}\right)\\
\mathrm{diag}\left(x^{m_i}- \lambda_i^{-1}\right)
\end{pmatrix}.
\end{equation}
Replacing $\mathbf{B}_1^T$ in ~\eqref{Th-Galois-Intersection-inproof1} with the polynomial matrix in~\eqref{Th-Galois-Intersection-inproof2}, we find that $\mathcal{Q}$ is generated by 
\begin{equation*}
\begin{pmatrix}
\begin{pmatrix}
\sigma^{e-\kappa}\left(\mathbf{G}_1\left(\frac{1}{x}\right) \mathrm{diag}\left(x^{m_i}\right) \right)\\
\mathrm{diag}\left(x^{m_i}-\delta_i\right)
\end{pmatrix}
\mathrm{diag}\left(\frac{x^N-1}{x^{m_i}-\delta_i}\right) \mathbf{G}_2^T\\
(x^N-1)\mathbf{I}_\ell
\end{pmatrix},
\end{equation*}
which shows that $\mathcal{Q}$ is equivalently generated by 
\begin{equation*}\begin{pmatrix}
\sigma^{e-\kappa}\left(\mathbf{G}_1\left(\frac{1}{x}\right) \mathrm{diag}\left(x^{m_i}\right) \right) \mathrm{diag}\left(\frac{x^N-1}{x^{m_i}-\delta_i}\right) \mathbf{G}_2^T\\
(x^N-1)\mathbf{I}_\ell
\end{pmatrix}.\end{equation*}
The result then follows directly from Theorem~\ref{Th-Intersection}.
\end{proof}

\begin{example}
\label{Code1-Code2-R5}
Consider the codes $\mathcal{C}_1$ and $\mathcal{C}_2$ presented in Example~\ref{Code1-Code2}. In Example~\ref{Code1-Code2-R4}, Corollary~\ref{Th-LinearGaloisIntersection} was applied to show that $\mathcal{C}_1^{\perp_1} \cap \mathcal{C}_2=\{\mathbf{0}\}$. In Example~\ref{Code1-Code2-R1}, $\mathcal{C}_1$ and $\mathcal{C}_2$ were shown to be $(1,\omega)$-MT codes with block lengths $(6,2)$ and GPMs 
\begin{equation*}
\mathbf{G}_1=\begin{pmatrix}
 \omega +x &  \omega \\
 0 & \omega^2+x
\end{pmatrix}\quad \text{and} \quad \mathbf{G}_2=\begin{pmatrix}
 \omega^2+ \omega^2 x + x^2+x^3  & \omega x \\
 0 & \omega+x^2
\end{pmatrix},
\end{equation*}
respectively. Furthermore, Example~\ref{Code1-Code2-R2} showed that the $1$-Galois dual $\mathcal{C}_1^{\perp_1}$ is also a $\left(1,\omega\right)$-MT code with block lengths $(6,2)$. Consequently, Corollary~\ref{Th-Galois-Intersection} implies that the intersection $\mathcal{C}_1^{\perp_1} \cap \mathcal{C}_2$ is a $\left(1,\omega\right)$-MT code with the same block lengths. In addition, this corollary provides a means to determine a GPM for this intersection. To this end, we consider the QC code $\mathcal{Q}$ over $\mathbb{F}_4$ of index $2$ and co-index $6$, generated by 
\begin{equation*}\begin{pmatrix}
\sigma\left(\mathbf{G}_1\left(\frac{1}{x}\right) \mathrm{diag}\left(x^{m_i}\right) \right)\mathrm{diag}\left(\frac{x^6-1}{x^{m_i}-\delta_i}\right) \mathbf{G}_2^T\\
(x^6-1)\mathbf{I}_2
\end{pmatrix}
=\begin{pmatrix}
\omega^2 x^3 + x^5 + x^6 + \omega x^7 + \omega x^8 +\omega^2 x^9 & \omega^2 x^2+\omega^2 x^8 \\
 x^2 + \omega x^3 + \omega^2 x^4 + x^5 + \omega x^6 +\omega^2 x^7 &  x+ \omega x^2 + x^7 +\omega x^8 \\
x^6-1 &0\\
0 & x^6-1
\end{pmatrix}.
\end{equation*}
By reducing to Hermite normal form, we obtain the reduced GPM for $\mathcal{Q}$, as well as the corresponding polynomial matrix that satisfies its identical equation. These are, respectively, given by
\begin{equation*}\mathbf{Q}=\begin{pmatrix}
x^3 + x^2 + \omega^2 x + \omega^2 & 0\\
0 & x^6 + 1
\end{pmatrix} \quad \text{and}\quad
\mathbf{P}=\begin{pmatrix}
 \omega +  \omega x + x^2 +x^3 & 0\\
 0 & 1
\end{pmatrix}.
\end{equation*}
In particular, it is evident that $\mathbf{P}\mathbf{U}=\mathbf{A}_2^T$, where 
\begin{equation*}
\mathbf{U}=\begin{pmatrix}
1 &  0  \\
\omega x +\omega x^2  & 1
\end{pmatrix}
\end{equation*}
is an invertible matrix, and $\mathbf{A}_2$ is defined in Example~\ref{Code1-Code2-R1}. As a result, $\mathbf{P}^T\mathbf{G}_2$ corresponds to the zero code. Alternatively, we observe from \eqref{Dim_Eq_fromA} and \eqref{Identical_Eq_QC} that $\dim{\mathcal{Q}}=\deg\left(\mathrm{det}\left( \mathbf{P}\right)\right)=\deg\left(\mathrm{det}\left( \mathbf{A}_2\right)\right)=\dim{\mathcal{C}_2}$. These two findings align with the results that will be presented in Lemma~\ref{Th-Containment}. Consequently, we have $\mathcal{C}^{\perp_1}_1\cap \mathcal{C}_2=\{\mathbf{0}\}$, consistent with Example~\ref{Code1-Code2-R4}.
\hfill $\diamond$
\end{example}

\section{Applications to MT codes intersection}
\label{Sec-applications}
This section is devoted for several applications of the theoretical results established in Sections \ref{Sec-MT-Reversed-Code} and \ref{Sec-Intersection-MT-Codes}. Specifically, we aim to establish necessary and sufficient conditions under which an MT code is Galois self-orthogonal, Galois dual-containing, Galois LCD, or reversible. These properties have been widely studied in the literature for certain subclasses of MT codes, and occasionally for MT codes in general. Nevertheless, the results presented in this section are significant for several reasons:
\begin{enumerate}
\item In literature, some prior conditions for self-orthogonal, dual-containing, or LCD MT codes have been shown to be incorrect in \cite{Takieldin2025}. Even if these conditions are corrected, they are sufficient but not necessary and apply only to special subclasses of MT codes. In contrast, the conditions we propose are necessary and sufficient and apply to general MT codes without restrictive assumptions, other than requiring the intersecting MT codes to have identical shift constants. This requirement is justified by Remark \ref{Imp-Remark}, which asserts the existence of unnecessary trivial blocks when this requirement is not met. 
\item Our proposed conditions rely entirely on the GPM defining the MT code, unlike some prior conditions in literature that necessitate decomposing the MT code into a direct sum of constituents of linear codes over various extension fields of $\mathbb{F}_q$.
\item To the best of our knowledge, this study is the first to systematically examine the reversibility in the class of MT codes.
\end{enumerate}

The main result of this section is presented in Theorem \ref{Th-SO-DC-R}, which establishes necessary and sufficient conditions for an MT code to be Galois self-orthogonal, Galois dual-containing, Galois LCD, or reversible. The condition of being Galois LCD, in particular, will be deduced as a consequence of the following result, where we examine a more general context. Namely, we provide a necessary and sufficient condition for the trivial intersection of two MT codes. This directly yields the condition for a Galois LCD MT code, since a Galois LCD code is equivalent to a trivial intersection between the code and its Galois dual. Although the following result requires the block lengths being coprime to $q$, Remark \ref{remark1} illustrates that this restriction may be eliminated with a minor adjustment. 

\begin{theorem}
\label{MT-LCP}
Let $\mathcal{C}_1$ and $\mathcal{C}_2$ be $\Lambda$-MT codes over $\mathbb{F}_q$ with index $\ell$, block lengths $(m_1, m_2, \ldots, m_\ell)$, and GPMs $\mathbf{G}_1$ and $\mathbf{G}_2$, respectively. Let $\mathbf{A}_1$ and $\mathbf{A}_2$ denote the corresponding polynomial matrices that satisfy the identical equation given by \eqref{Identical_Eq} for $\mathbf{G}_1$ and $\mathbf{G}_2$. Assume that each $m_i$ is coprime to $q$, and consider the factorization of the polynomial $x^N - 1$ over $\mathbb{F}_q$ as
$$x^N-1=\prod_{j=1}^{s}p_j(x),$$
where $N=\mathrm{lcm}(t_i m_i)$ and each $p_j(x)$ is an irreducible polynomial in $\mathbb{F}_q[x]$ for $1 \le j \le s$. Then, $\mathcal{C}_1 \cap \mathcal{C}_2 = \{\mathbf{0}\}$ if and only if
$$\sum_{j=1}^{s} r_j \deg\left(p_j(x)\right)=\dim\left(\mathcal{C}_2\right),$$
where $r_j$ is the rank of the matrix
$$\mathbf{G}_2 \ \mathrm{diag}\left(\frac{x^N-1}{x^{m_i}-\lambda_i}\right) \mathbf{A}_1 \pmod{p_j(x)}.$$
\end{theorem}
\begin{proof} 
The assumption that each $m_i$ is coprime to $q$ implies that $N$ is also coprime to $q$. This ensures that $x^N - 1$ factors into distinct irreducible polynomials $p_j(x)$. From Equations \eqref{Dim_Eq_fromA} and \eqref{Identical_Eq_QC}, along with the same notation used in Theorem \ref{Th-Intersection}, we have 
\begin{equation}
\label{MT-LCP-Inproof}
\begin{split}
\dim\left(\mathcal{Q}\right)=\deg{\mathrm{det}}\left(\mathbf{P}\right)&=\deg{\mathrm{det}}\left(\mathbf{P}^T\mathbf{G}_2\right)-\deg{\mathrm{det}}\left(\mathbf{G}_2\right)\\
&=n-\dim\left(\mathcal{C}_1\cap \mathcal{C}_2\right)-\left(n-\dim\left(\mathcal{C}_2\right)\right)=\dim\left(\mathcal{C}_2\right)-\dim\left(\mathcal{C}_1\cap \mathcal{C}_2\right).
\end{split}
\end{equation}
The discussion at the beginning of Section \ref{Sec-MT-Reversed-Code} shows that the QC code $\mathcal{Q}$ can be regarded as a linear code of length $\ell$ over the quotient ring $R = \mathbb{F}_q[x]/\langle x^N - 1 \rangle$. Meanwhile, a generator matrix for $\mathcal{Q}$ as a code over $R$ is provided by Theorem \ref{Th-Intersection} and is given by
$$\mathbf{A}_1^T \mathrm{diag}\left(\frac{x^N-1}{x^{m_i}-\lambda_i}\right) \mathbf{G}_2^T \pmod{x^N-1}.$$
According to \cite{Ling2001}, $\mathcal{Q}$ decomposes into a direct sum of linear codes $\mathcal{Q}_j$, each of length $\ell$, over the finite field $R_j=\mathbb{F}_q[x]/\langle p_j\left(x\right)\rangle$. For each $1\le j\le s$, the code $\mathcal{Q}_j$ has a generator matrix over $R_j$ given by
$$\mathbf{A}_1^T \mathrm{diag}\left(\frac{x^N-1}{x^{m_i}-\lambda_i}\right) \mathbf{G}_2^T \pmod{p_j(x)}.$$
Since the rank of this generator matrix is $r_j$, and the order of $R_j$ is $q^{\deg\left(p_j(x)\right)}$, it follows that $|\mathcal{Q}_j|=q^{r_j \deg\left(p_j(x)\right)}$. Therefore, using \eqref{MT-LCP-Inproof}, we have
\begin{equation}
\label{MT-LCP-Inproof2}
\begin{split}
\dim\left(\mathcal{C}_2\right)-\dim\left(\mathcal{C}_1\cap \mathcal{C}_2\right)&=\dim\left(\mathcal{Q}\right)=\log_q |\mathcal{Q}| =\log_q \prod_{j=1}^{s} |\mathcal{Q}_j| \\&=\sum_{j=1}^{s} \log_q |\mathcal{Q}_j| =\sum_{j=1}^{s} \log_q\left( q^{r_j \deg\left(p_j(x)\right)} \right)=\sum_{j=1}^{s} r_j \deg\left(p_j(x)\right).
\end{split}
\end{equation}
It follows that $\mathcal{C}_1 \cap \mathcal{C}_2 =\{\mathbf{0}\}$ if and only if $\sum_{j=1}^{s} r_j \deg\left(p_j(x)\right)=\dim\left(\mathcal{C}_2\right)$. 
\end{proof}

\begin{remark}
\label{remark1}
The case in which at least one $m_i$ is not coprime to $q$ is indeed possible, as shown in the next example. In fact, removing the assumption that each $m_i$ is coprime to $q$ from Theorem \ref{MT-LCP} necessitates a slight adjustment in its result. In this more general setting, the polynomial $x^N - 1$ factors into irreducible polynomials that are not necessarily distinct. That is,
$$x^N-1=\prod_{j=1}^{s} p_j^{f_j}(x),$$
where $p_j(x)$ are distinct irreducible polynomials and $f_j \ge 1$ for $1 \le j \le s$. In this case, $\mathcal{Q}$ decomposes into a direct sum of linear codes $\mathcal{Q}_j$, each of length $\ell$, over the finite chain ring $R_j=\mathbb{F}_q[x]/\langle p_j^{f_j}\left(x\right)\rangle$. Each $\mathcal{Q}_j$ as a code over $R_j$ is generated by the matrix
\begin{equation}
\label{Eq-in-remark1}
\mathbf{A}_1^T \mathrm{diag}\left(\frac{x^N-1}{x^{m_i}-\lambda_i}\right) \mathbf{G}_2^T \pmod{p_j^{f_j}(x)}.
\end{equation}
According to Remark \ref{Rem_Chain}, if $\mathcal{Q}_j$ is of type $\left\{r_0^{(j)}, r_1^{(j)}, \ldots, r_{f_j-1}^{(j)}\right\}$, then \eqref{Eq_Chain} yields
$$|\mathcal{Q}_j|=q^{\deg\left(p_j(x)\right)\sum_{h=0}^{f_j-1}(f_j-h) r_h^{(j)}}.$$ 
Following steps analogous to those in the proof of Theorem \ref{MT-LCP}, Equation \eqref{MT-LCP-Inproof2} takes the form
\begin{equation*}
\dim\left(\mathcal{C}_2\right)-\dim\left(\mathcal{C}_1\cap \mathcal{C}_2\right)=\sum_{j=1}^{s} \deg\left(p_j(x)\right)\sum_{h=0}^{f_j-1}(f_j-h) r_h^{(j)}.
\end{equation*}
Consequently, $\mathcal{C}_1 \cap \mathcal{C}_2 =\{\mathbf{0}\}$ if and only if 
$$\sum_{j=1}^{s} \deg\left(p_j(x)\right)\sum_{h=0}^{f_j-1}(f_j-h) r_h^{(j)}=\dim\left(\mathcal{C}_2\right).$$ 
\hfill $\diamond$
\end{remark}

\begin{example}
\label{Code1-Code2-R7}
Consider the $(1, \omega)$-MT codes $\mathcal{C}_1^{\perp_1}$ and $\mathcal{C}_2$ over $\mathbb{F}_4$ with block lengths $(6, 2)$, as presented in Examples \ref{Code1-Code2-R1} and \ref{Code1-Code2-R2}. Since $N=\mathrm{lcm}(t_i m_i)=6$, the polynomial $x^6 - 1$ factors over $\mathbb{F}_4$ as
$$x^6-1=(x + 1)^2  (x + \omega)^2  (x + \omega^2)^2.$$
Define $p_1 = x + 1$, $p_2 = x + \omega$, $p_3 = x + \omega^2$, and $f_1 = f_2 = f_3 = 2$, since the block lengths are not coprime to $q$. Example \ref{Code1-Code2-R2} provides the GPM $\sigma \left(\mathbf{H}_1\right)$ for $\mathcal{C}_1^{\perp_1}$, along with the corresponding polynomial matrix $\sigma\left(\mathbf{B}_1\right)$ that satisfies its identical equation. On the other hand, Example \ref{Code1-Code2-R1} provides $\mathbf{G}_2$ as a GPM for $\mathcal{C}_2$. Recall that
\begin{equation*}
\sigma\left(\mathbf{B}_1\right)=\begin{pmatrix}
1+x+x^2  & \omega+x \\
0 &  1
\end{pmatrix}
\quad \text{and} \quad \mathbf{G}_2=\begin{pmatrix}
 \omega^2+ \omega^2 x + x^2+x^3  & \omega x \\
 0 & \omega+x^2
\end{pmatrix}.
\end{equation*}
Equation \eqref{Eq-in-remark1} requires computing
\begin{equation*}\begin{split}
\sigma\left(\mathbf{B}_1\right)^T \mathrm{diag}\left(\frac{x^6-1}{x^{m_i}-\lambda_i}\right) \mathbf{G}_2^T 
&=\begin{pmatrix}
(x + 1)   (x + \omega)^3 (x + \omega^2)  & 0\\
 \omega (x + 1)  (x + \omega)^2  (x^2 + \omega x + 1) & x^6-1
\end{pmatrix}.
\end{split}\end{equation*}
Reductions modulo $p_j^{f_j}(x)$ shows that $\mathcal{Q}_1$ is of type $\{0,1\}$, $\mathcal{Q}_2$ is of type $\{0,0\}$, and $\mathcal{Q}_3$ is of type $\{1,0\}$. Therefore, 
$$\sum_{j=1}^{s} \deg\left(p_j(x)\right)\sum_{h=0}^{f_j-1}(f_j-h) r_h^{(j)}=3=\dim\left(\mathcal{C}_2\right).$$ 
By Remark \ref{remark1}, this confirms that $\mathcal{C}^{\perp_1}_1\cap \mathcal{C}_2=\{\mathbf{0}\}$, which is consistent with the results in Examples \ref{Code1-Code2-R4} and \ref{Code1-Code2-R5}.
\hfill $\diamond$
\end{example}

Alternative conditions to those in Theorem \ref{MT-LCP} and Remark \ref{remark1} under which two MT codes intersect trivially are given in the first part of the following result. The second part is more significant since it is the first step toward proving the conditions under which an MT code is Galois self-orthogonal, Galois dual-containing, or reversible, as will be shown in Theorem \ref{Th-SO-DC-R}. Recall that a code is Galois self-orthogonal if it is contained in its Galois dual, while it is Galois dual-containing if it contains its Galois dual. We start with a more general context that provides a condition under which one MT code $\mathcal{C}_1$ is contained in another $\mathcal{C}_2$. In light of Remark \ref{Imp-Remark}, we assume that $\mathcal{C}_1$ and $\mathcal{C}_2$ are both $\Lambda$-MT. This is due to the following reason: Let $\mathcal{C}_1$ and $\mathcal{C}_2$ be $\Delta$-MT and $\Lambda$-MT, respectively, and suppose that $d\left(\mathcal{C}_2\right) > D\left(\Lambda-\Delta\right)$, where $D\left(\Lambda-\Delta\right)$ denotes the number of indices at which $\Lambda$ and $\Delta$ differ. We aim to propose a condition under which $\mathcal{C}_1 \subseteq \mathcal{C}_2$, or equivalently, $\mathcal{C}_1 \cap \mathcal{C}_2 = \mathcal{C}_1$. This means that the intersection admits an MT structure. By a similar reasoning to that used in the proof of Corollary \ref{Corr-Th-CondOfInter}, this implies that the projection of $\mathcal{C}_1$ onto its $i$-th block must be zero for every index $i$ at which $\Lambda$ and $\Delta$ differ. Consequently, $\mathcal{C}_1$ is also $\Lambda$-MT. Therefore, we may assume from the beginning that both $\mathcal{C}_1$ and $\mathcal{C}_2$ are $\Lambda$-MT codes. This assumption moreover eliminates the need to impose any conditions on the minimum distances of the codes.

\begin{lemma}
\label{Th-Containment}
Using the same notation as in Theorem \ref{Th-Intersection}, the following statements are equivalent: 
\begin{enumerate}
\item $\mathcal{C}_1 \cap \mathcal{C}_2=\{\mathbf{0}\}$.
\item The columns of $\mathbf{A}_1$ and $\mathbf{A}_2$ together generate the entire module $\left(\mathbb{F}_q[x]\right)^\ell$.
\item $\mathrm{diag}\left(\frac{x^N-1}{x^{m_i}-\lambda_i}\right) \mathbf{G}_2^T$ is a GPM for $\mathcal{Q}$.
\item There exists an invertible polynomial matrix $\mathbf{U}$ such that $\mathbf{P}\mathbf{U}=\mathbf{A}_2^T$.
\item $\dim\left(\mathcal{Q}\right)=\dim\left(\mathcal{C}_2\right)$.
\end{enumerate}
On the other hand, the inclusion $\mathcal{C}_1 \subseteq \mathcal{C}_2$ holds if and only if 
\begin{equation*}
\mathbf{G}_1 \mathrm{diag}\left(\frac{x^N-1}{x^{m_i}-\lambda_i}\right) \mathbf{A}_2 \equiv \mathbf{0} \pmod{x^N-1}.
\end{equation*}
\end{lemma}
\begin{proof} 
The notation follows that used in the proof of Theorem \ref{Th-Intersection}. Suppose that $\mathcal{C}_1 \cap \mathcal{C}_2 = \{\mathbf{0}\}$, then $\overline{\mathbf{G}}=\mathrm{diag}(x^{m_i}-\lambda_i)$. Consequently, we have $\overline{\mathbf{A}}=\mathbf{I}_\ell$, and thus $\mathbf{A}_1\mathbf{M}_1+\mathbf{A}_2\mathbf{M}_2=\mathbf{I}_\ell$ from \eqref{Th-Intersection-InProof}. This implies that the columns of $\mathbf{A}_1$ and $\mathbf{A}_2$ generate $\left(\mathbb{F}_q[x]\right)^\ell$. Therefore, the GPM for $\mathcal{Q}$ given in the proof of Theorem \ref{Th-Intersection} is 
$$\overline{\mathbf{A}}^T \mathrm{diag}\left(\frac{x^N-1}{x^{m_i}-\lambda_i}\right) \mathbf{G}_2^T=\mathrm{diag}\left(\frac{x^N-1}{x^{m_i}-\lambda_i}\right) \mathbf{G}_2^T.$$
If $\mathbf{Q}$ denotes another GPM for $\mathcal{Q}$, then it must take the form $\mathbf{Q}=\mathbf{U}\mathrm{diag}\left(\frac{x^N-1}{x^{m_i}-\lambda_i}\right) \mathbf{G}_2^T$, for some invertible $\mathbf{U}$. From \eqref{Eq-for-proof}, we have
$$\mathbf{P}\mathbf{U}\mathrm{diag}\left(\frac{x^N-1}{x^{m_i}-\lambda_i}\right) \mathbf{G}_2^T=\mathbf{P}\mathbf{Q}=(x^N-1)\mathbf{I}_\ell=\mathbf{A}_2^T \mathrm{diag}\left(\frac{x^N-1}{x^{m_i}-\lambda_i}\right) \mathbf{G}_2^T,$$ 
which implies $\mathbf{P}\mathbf{U} =\mathbf{A}_2^T $. According to Theorem \ref{Th-Intersection}, a GPM for $\mathcal{C}_1 \cap \mathcal{C}_2$ is given by $\mathbf{P}^T \mathbf{G}_2$, or equivalently, $\mathbf{A}_2 \mathbf{G}_2=\mathrm{diag}\left( x^{m_i}-\lambda_i\right)$, showing that $\mathcal{C}_1 \cap \mathcal{C}_2=\{\mathbf{0}\}$. Finally, from \eqref{MT-LCP-Inproof}, $\mathcal{C}_1 \cap \mathcal{C}_2=\{\mathbf{0}\}$ holds if and only if $\dim\left(\mathcal{Q}\right)=\dim\left(\mathcal{C}_2\right)$.

The second assertion is straightforward to verify since the congruence $\mathbf{G}_1 \mathrm{diag}\left(\frac{x^N-1}{x^{m_i}-\lambda_i}\right) \mathbf{A}_2 \equiv \mathbf{0} \pmod{x^N-1}$ holds if and only if there exists a polynomial matrix $\mathbf{M}$ such that
\begin{equation*}
\mathbf{G}_1 \mathrm{diag}\left(\frac{x^N-1}{x^{m_i}-\lambda_i}\right) \mathbf{A}_2 =\mathbf{M}\left( x^N-1\right)=\mathbf{M} \mathbf{G}_2 \mathrm{diag}\left(\frac{x^N-1}{x^{m_i}-\lambda_i}\right) \mathbf{A}_2.
\end{equation*}
This is equivalent to $\mathbf{G}_1=\mathbf{M} \mathbf{G}_2$, or, in other words, $\mathcal{C}_1 \subseteq \mathcal{C}_2$. 
\end{proof}

The following theorem constitutes the main result of this section. By using Theorem \ref{MT-LCP} and Lemma \ref{Th-Containment}, we derive necessary and sufficient conditions under which an MT code is Galois self-orthogonal, Galois dual-containing, Galois LCD, or reversible.

\begin{theorem}
\label{Th-SO-DC-R}
Let $\mathcal{C}$ be a $\Lambda$-MT code over $\mathbb{F}_{p^e}$ with index $\ell$, block lengths $(m_1,m_2,\ldots,m_\ell)$, GPM $\mathbf{G}$, and a polynomial matrix $\mathbf{A}$ satisfying the identical equation given by \eqref{Identical_Eq} for $\mathbf{G}$. Set $N=\mathrm{lcm}(t_i m_i)$.
\begin{enumerate}
\item Suppose that $\sigma^{e-\kappa}\left(\Lambda^{-1}\right)=\Lambda$ for some integer $0\le \kappa <e$. Then,
\begin{enumerate}
\item $\mathcal{C}$ is $\kappa$-Galois self-orthogonal if and only if 
$$\sigma^{e-\kappa}\left(\mathbf{G}\left(\frac{1}{x}\right) \mathrm{diag}\left(x^{m_i}\right) \right) \mathrm{diag}\left(\frac{x^N-1}{x^{m_i}-\lambda_i}\right) \mathbf{G}^T \equiv \mathbf{0} \pmod{x^N-1}.$$
\item $\mathcal{C}$ is $\kappa$-Galois dual containing if and only if 
$$\sigma^{e-\kappa}\left(\mathbf{A}^T\left(\frac{1}{x}\right) \mathrm{diag}\left(x^{m_i}\right) \right) \mathrm{diag}\left(\frac{x^N-1}{x^{m_i}-\lambda_i}\right) \mathbf{A} \equiv \mathbf{0} \pmod{x^N-1}.$$
\item Assume that each $m_i$ is coprime to $q$, and consider the factorization of the polynomial $x^N - 1$ over $\mathbb{F}_q$ to irreducible polynomials as $x^N-1=\prod_{j=1}^{s}p_j(x)$. Then $\mathcal{C}$ is $\kappa$-Galois LCD if and only if 
$$\sum_{j=1}^{s} r_j \deg\left(p_j(x)\right)=\dim\left(\mathcal{C}\right),$$
where $r_j$ is the rank of the matrix
$$\sigma^{e-\kappa}\left(\mathbf{G}\left(\frac{1}{x}\right) \mathrm{diag}\left(x^{m_i}\right) \right)\mathrm{diag}\left(\frac{x^N-1}{x^{m_i}-\lambda_i}\right) \mathbf{G}^T \pmod{p_j(x)}.$$
\end{enumerate}

\item Suppose that $m_i=m_{\ell-i+1}$ and $\lambda_i= \lambda_{\ell-i+1}^{-1}$ for $1\le i\le \ell$. Then, $\mathcal{C}$ is reversible if and only if 
$$\mathbf{G}\left(\frac{1}{x}\right) \mathrm{diag}\left(x^{m_i}\right)  J_\ell \ \mathrm{diag}\left(\frac{x^N-1}{x^{m_i}-\lambda_i}\right) \mathbf{A} \equiv \mathbf{0} \pmod{x^N-1}.$$
\end{enumerate}
\end{theorem}

\begin{proof}
\begin{enumerate}
\item The assumption $\sigma^{e-\kappa}\left(\Lambda^{-1}\right)=\Lambda$ implies that $\mathcal{C}^{\perp_\kappa}$ is also a $\Lambda$-MT code with block lengths $(m_1,m_2,\ldots,m_\ell)$. 
\begin{enumerate}
\item Consider Corollary \ref{Th-Galois-Intersection} with $\mathcal{C}_1 = \mathcal{C}_2 = \mathcal{C}$. In this case, $\mathbf{P}^T \mathbf{G}$ is a GPM for the $\Lambda$-MT code $\mathcal{C}^{\perp_\kappa} \cap \mathcal{C}$. Therefore, $\mathcal{C}$ is $\kappa$-Galois self-orthogonal if and only if $\mathbf{P}^T \mathbf{G}$ is a GPM for $\mathcal{C}$, which holds if and only if $\mathbf{P}$ is invertible. By \eqref{Identical_Eq_QC}, this is equivalent to
$$\mathbf{Q}=(x^N-1)\mathbf{P}^{-1}\equiv \mathbf{0} \pmod{x^N-1}.$$
Corollary \ref{Th-Galois-Intersection} shows that this condition holds if and only if
$$\sigma^{e-\kappa}\left(\mathbf{G}\left(\frac{1}{x}\right) \mathrm{diag}\left(x^{m_i}\right) \right) \mathrm{diag}\left(\frac{x^N-1}{x^{m_i}-\lambda_i}\right) \mathbf{G}^T\equiv \mathbf{0} \pmod{x^N-1}.$$
We remark that this result may alternatively be proved by applying Lemma \ref{Th-Containment} with $\mathcal{C}_1 = \mathcal{C}$ and $\mathcal{C}_2 = \mathcal{C}^{\perp_\kappa}$, which yields $\mathbf{A}_2 = \sigma^{e-\kappa}\left(\mathbf{B}\right)$. Consequently, $\mathcal{C}$ is $\kappa$-Galois self-orthogonal (i.e., $\mathcal{C} \subseteq  \mathcal{C}^{\perp_\kappa}$) if and only if 
$$\sigma^{e-\kappa}\left(\mathbf{B}^T\right) \mathrm{diag}\left(\frac{x^N-1}{x^{m_i}-\lambda_i}\right) \mathbf{G}^T \equiv \mathbf{0} \pmod{x^N-1}.$$
According to Lemma \ref{Th-Reciprocal}, $\mathbf{B}^T$ is a GPM for $\mathcal{L}\left(\mathcal{C}\right)$, and thus it can be replaced by the polynomial matrix given in \eqref{Inproof_Th-Reciprocal}.
\item Consider Lemma \ref{Th-Containment} with $\mathcal{C}_1=\mathcal{C}^{\perp_\kappa}$ and $\mathcal{C}_2=\mathcal{C}$. In this setting, we take $\mathbf{G}_1=\sigma^{e-\kappa}\left(\mathbf{H} \right)$ and $\mathbf{A}_2=\mathbf{A}$. According to \eqref{Eq-Reciprocal1}, $\sigma^{e-\kappa}\left(\mathbf{H} \right)$ is a GPM for $\mathcal{C}^{\perp_\kappa}$ and can be replaced by the polynomial matrix
\begin{equation*}
\sigma^{e-\kappa}\left(\begin{pmatrix}
\mathbf{A}^T\left(\frac{1}{x}\right) \mathrm{diag}\left(x^{m_i}\right)\\
\mathrm{diag}\left(x^{m_i}- \lambda_i^{-1} \right)
\end{pmatrix}\right)= \begin{pmatrix}
\sigma^{e-\kappa}\left(\mathbf{A}^T\left(\frac{1}{x}\right) \mathrm{diag}\left(x^{m_i}\right)\right)\\
\mathrm{diag}\left(x^{m_i}- \lambda_i \right)
\end{pmatrix}.
\end{equation*}
Lemma \ref{Th-Containment} shows that $\mathcal{C}$ is $\kappa$-Galois dual containing (i.e., $\mathcal{C}^{\perp_\kappa}\subseteq \mathcal{C}$) if and only if 
$$\begin{pmatrix}
\sigma^{e-\kappa}\left(\mathbf{A}^T\left(\frac{1}{x}\right) \mathrm{diag}\left(x^{m_i}\right)\right)\\
\mathrm{diag}\left(x^{m_i}- \lambda_i \right)
\end{pmatrix} \mathrm{diag}\left(\frac{x^N-1}{x^{m_i}-\lambda_i}\right) \mathbf{A} \equiv \mathbf{0}\pmod{x^N-1}.$$
This condition is equivalent to $\sigma^{e-\kappa}\left(\mathbf{A}^T\left(\frac{1}{x}\right) \mathrm{diag}\left(x^{m_i}\right)\right) \mathrm{diag}\left(\frac{x^N-1}{x^{m_i}-\lambda_i}\right) \mathbf{A} \equiv \mathbf{0}\pmod{x^N-1}$.
\item Consider Theorem \ref{MT-LCP} with $\mathcal{C}_1=\mathcal{C}^{\perp_\kappa}$ and $\mathcal{C}_2=\mathcal{C}$. In this setting, we take $\mathbf{A}_1=\sigma^{e-\kappa}\left(\mathbf{B} \right)$ and $\mathbf{G}_2=\mathbf{G}$. Therefore, $\mathcal{C}$ is $\kappa$-Galois LCD (i.e., $\mathcal{C}^{\perp_\kappa} \cap \mathcal{C}=\{\mathbf{0}\}$) if and only if $\sum_{j=1}^{s} r_j \deg\left(p_j(x)\right)=\dim\left(\mathcal{C}\right)$, where $r_j$ is the rank of 
$$\sigma^{e-\kappa}\left(\mathbf{B}^T \right) \mathrm{diag}\left(\frac{x^N-1}{x^{m_i}-\lambda_i}\right) \mathbf{G}^T \pmod{p_j(x)}.$$
Again, $\mathbf{B}^T$ is a GPM for $\mathcal{L}\left(\mathcal{C}\right)$ by Lemma \ref{Th-Reciprocal}, and can be replaced by the polynomial matrix given in \eqref{Inproof_Th-Reciprocal}.
\end{enumerate}
\item According to Theorem \ref{Th-Reversed}, the assumptions $m_i=m_{\ell-i+1}$ and $\lambda_i= \lambda_{\ell-i+1}^{-1}$ imply that $\mathcal{R}$ is also a $\Lambda$-MT code with block lengths $(m_1,m_2,\ldots,m_\ell)$. Since $\dim\left(\mathcal{C}\right)=\dim\left(\mathcal{R}\right)$, it follows that $\mathcal{C}$ is reversible if and only if $\mathcal{R}\subseteq \mathcal{C}$. To verify this inclusion, we apply Lemma \ref{Th-Containment} with $\mathcal{C}_1=\mathcal{R}$ and $\mathcal{C}_2=\mathcal{C}$. From Theorem \ref{Th-Reversed}, we use $\mathbf{G}_1=\mathbf{B}^T J_\ell$ and $\mathbf{A}_2=\mathbf{A}$. By Lemma \ref{Th-Reciprocal}, $\mathbf{B}^T J_\ell$ can be replaced by the polynomial matrix 
\begin{equation*}
\begin{pmatrix}
\mathbf{G}\left(\frac{1}{x}\right) \mathrm{diag}\left(x^{m_i}\right)\\
\mathrm{diag}\left(x^{m_i}- \lambda_i^{-1}\right)
\end{pmatrix}  J_\ell =
\begin{pmatrix}
\mathbf{G}\left(\frac{1}{x}\right) \mathrm{diag}\left(x^{m_i}\right)  J_\ell\\
\mathrm{diag}\left(x^{m_i}- \lambda_i \right)
\end{pmatrix} .\end{equation*}
Applying Lemma \ref{Th-Containment}, we conclude that $\mathcal{C}$ is reversible (i.e., $\mathcal{R}\subseteq \mathcal{C}$) if and only if 
$$\begin{pmatrix}
\mathbf{G}\left(\frac{1}{x}\right) \mathrm{diag}\left(x^{m_i}\right)  J_\ell\\
\mathrm{diag}\left(x^{m_i}- \lambda_i \right)
\end{pmatrix} \mathrm{diag}\left(\frac{x^N-1}{x^{m_i}-\lambda_i}\right) \mathbf{A} \equiv \mathbf{0}\pmod{x^N-1}.$$
This condition reduces to $\mathbf{G}\left(\frac{1}{x}\right) \mathrm{diag}\left(x^{m_i}\right)  J_\ell \ \mathrm{diag}\left(\frac{x^N-1}{x^{m_i}-\lambda_i}\right) \mathbf{A} \equiv \mathbf{0}\pmod{x^N-1}$.
\end{enumerate}
 \end{proof}

The condition regarding Galois LCD MT codes established in Theorem \ref{Th-SO-DC-R} necessitates that the block lengths be coprime to $q$. The following remark illustrates that this restriction can be lifted with a minor adjustment to the condition.

\begin{remark}
\label{remark2}
Suppose that at least one of the integers $m_i$ is not coprime to $q$. Under this assumption, $x^N - 1$ factors into irreducible polynomials that are not necessarily distinct. Specifically,
$$x^N-1=\prod_{j=1}^{s} p_j^{f_j}(x),$$
where $p_j(x)$ are distinct irreducible polynomials and $f_j \ge 1$ for $1 \le j \le s$. Consider Remark \ref{remark1} with $\mathcal{C}_1=\mathcal{C}^{\perp_\kappa}$ and $\mathcal{C}_2=\mathcal{C}$. In this setting, we take $\mathbf{A}_1=\sigma^{e-\kappa}\left(\mathbf{B} \right)$ and $\mathbf{G}_2=\mathbf{G}$. Equation \eqref{Eq-in-remark1} takes the following form after replacing $\mathbf{B}^T$ by the polynomial matrix given in \eqref{Inproof_Th-Reciprocal}:
\begin{equation}
\label{Eq-in-remark2}
\sigma^{e-\kappa}\left(\mathbf{G}\left(\frac{1}{x}\right) \mathrm{diag}\left(x^{m_i}\right) \right) \mathrm{diag}\left(\frac{x^N-1}{x^{m_i}-\lambda_i}\right) \mathbf{G}^T \pmod{p_j^{f_j}(x)}
\end{equation}
For each $1 \leq j \leq s$, let $\mathcal{Q}_j$ be the code over the finite chain ring $R_j=\mathbb{F}_q[x]/\langle p_j^{f_j}\left(x\right)\rangle$ that is generated by the matrix given in \eqref{Eq-in-remark2}. According to Remark \ref{Rem_Chain}, if $\mathcal{Q}_j$ is of type $\left\{r_0^{(j)}, r_1^{(j)}, \ldots, r_{f_j-1}^{(j)}\right\}$, then Remark \ref{remark1} implies that $\mathcal{C}$ is $\kappa$-Galois LCD (i.e., $\mathcal{C}^{\perp_\kappa} \cap \mathcal{C}=\{\mathbf{0}\}$) if and only if 
$$\sum_{j=1}^{s} \deg\left(p_j(x)\right)\sum_{h=0}^{f_j-1}(f_j-h) r_h^{(j)}=\dim\left(\mathcal{C}\right).$$
\hfill $\diamond$
\end{remark}

This section concludes with some examples that illustrate the application of Theorem \ref{Th-SO-DC-R}. It is worth noting that, for a QC code with index $\ell$ and co-index $m$, the assumptions stated in Theorem \ref{Th-SO-DC-R} are satisfied, e.g., $\sigma^{e-\kappa}\left(\Lambda^{-1}\right)=\Lambda$ for any $\kappa$. Moreover, we have
$$N=m, \  \mathrm{diag}\left(x^{m_i}\right) =x^m, \ \text{and}\ \mathrm{diag}\left(\frac{x^N-1}{x^{m_i}-\lambda_i}\right)=\mathbf{I}_\ell.$$

\begin{example}
\label{Code3-Code4-R2}
Consider the QC code $\mathcal{C}_3$ with index $\ell=3$, co-index $m=3$, and polynomial matrices $\mathbf{G}_3$ and $\mathbf{A}_3$, as given in Example \ref{Code3-Code4-R1}. According to Theorem \ref{Th-SO-DC-R},  $\mathcal{C}_3$ is Euclidean self-orthogonal; that is, $\kappa$-Galois self-orthogonal with $\kappa=0$. This follows from the congruence
\begin{equation*}
\begin{split}
\mathbf{G}_3\left(\frac{1}{x}\right) x^3 \mathbf{G}_3^T&=(x^3-1)\begin{pmatrix}
 x +2 x^2 &  2 x + x^2 +x^3 & x^2 + x^3\\
 2+ 2 x +x^2  & 1+2 x^3  & 0\\
2+2 x  & 0 &  1+2 x^3 
\end{pmatrix}\\
&\equiv \mathbf{0} \pmod{x^3-1}.
\end{split}
\end{equation*}

In addition, Theorem \ref{Th-SO-DC-R} shows that $\mathcal{C}_3$ is not reversible, as evidenced by the congruence
\begin{equation*}
\begin{split}
\mathbf{G}_3\left(\frac{1}{x}\right) x^3  J_3 \mathbf{A}_3&=\begin{pmatrix}
2 x^2+ 2 x^3 + x^5 +x^6  & 2 x + 2 x^3 +x^5 & 2 x^2+ 2 x^3 +2 x^4 \\
0 & 1+2 x^3  & 0\\
2 + 2 x^3 +2 x^6 & 2+ 2 x + x^2 + x^3 + x^4 +2 x^5  &    2+ 2 x + x^3+x^4
\end{pmatrix}\\
&\equiv\begin{pmatrix}
0  & 2+ 2 x +x^2  & 2+ 2x +2 x^2 \\
0&0&0\\
0&0&0
\end{pmatrix}
 \not \equiv \mathbf{0} \pmod{x^3-1}.
\end{split}\end{equation*}
This last result coincides with what Theorem \ref{reversible-subcode} showed in Example \ref{Code3-Code4}.
\hfill $\diamond$
\end{example}

\begin{example}
\label{Code3-Code4-R3}
Consider the $(2,1,2)$-MT code $\mathcal{C}_4$ with index $\ell=3$, block lengths $(3,3,3)$, and polynomial matrices $\mathbf{G}_4$ and $\mathbf{A}_4$, as given in Example \ref{Code3-Code4-R1}. According to Theorem \ref{Th-SO-DC-R}, $\mathcal{C}_4$ is Euclidean dual-containing; that is, $\kappa$-Galois dual-containing with $\kappa=0$. This follows from the congruence
$$\mathbf{A}_4^T\left(\frac{1}{x}\right) x^3  \mathrm{diag}\left(\frac{x^6-1}{x^{3}-\lambda_i}\right) \mathbf{A}_4 \equiv \mathbf{0} \pmod{x^6-1}.$$

In addition, Theorem \ref{Th-SO-DC-R} shows that $\mathcal{C}_4$ is reversible, as evidenced by the congruence
$$\mathbf{G}_4\left(\frac{1}{x}\right)  x^3  J_3 \ \mathrm{diag}\left(\frac{x^6-1}{x^{3}-\lambda_i}\right) \mathbf{A}_4 \equiv \mathbf{0} \pmod{x^6-1}.$$ 
This last result coincides with what Theorem \ref{reversible-subcode} showed in Example \ref{Code3-Code4}.
\hfill $\diamond$
\end{example}

\begin{example}
\label{Code6}
Consider the linear $[16,5,5]$ code $\mathcal{C}_6$ over $\mathbb{F}_9$ with generator matrix 
\begin{equation*}
G_6=\begin{pmatrix}
1&0&\omega^6&\omega&\omega^2&1&\omega^6&2&\omega^2&0&0&0&\omega&\omega^6&1&\omega\\
0&1&\omega^3&\omega^2&\omega^2&1&\omega^6&2&\omega^2&0&0&0&\omega^7&2&\omega^6&\omega^7\\
0&0&0&0&0&0&0&0&0&1&0&0&2&\omega^7&1&\omega^5\\
0&0&0&0&0&0&0&0&0&0&1&0&\omega&1&2&\omega^7\\
0&0&0&0&0&0&0&0&0&0&0&1&\omega^3&2&\omega&1
\end{pmatrix},
\end{equation*}
where $\omega \in \mathbb{F}_9$ satisfies $\omega^2+2\omega +2=0$. It can be verified that $\mathcal{C}_6$ is a $(1, \omega^2, 2)$-MT code with block lengths $(4, 5, 7)$, and its reduced GPM is given by
\begin{equation*}
\mathbf{G}_6=\begin{pmatrix}
\omega^6 + \omega x + x^2 & \omega^6+ 2 x +  \omega^2 x^2 +  x^3 + \omega^6 x^4  &  \omega^5 +  \omega^2 x +  2 x^2 +\omega^5 x^3 \\
0 & x^5 - \omega^2 & 0\\
0 & 0 & 1+ \omega^3 x + 2 x^2 + \omega x^3  + x^4 
\end{pmatrix}.
\end{equation*} 
Consequently, $N=\mathrm{lcm}(t_i m_i)=140$. Let $r_j$ denote the rank of the matrix described in Theorem \ref{Th-SO-DC-R}, namely
\begin{equation*}
\sigma \left(\mathbf{G}_6\left(\frac{1}{x}\right) \mathrm{diag}\left(x^{m_i}\right) \right)\mathrm{diag}\left(\frac{x^{140}-1}{x^{m_i}-\lambda_i}\right) \mathbf{G}_6^T \pmod{p_j(x)}.
\end{equation*}
Through tremendously computations, it is determined that $r_j=1$ for the irreducible factors $p_j(x)=x+1$, $p_j(x)=x + \omega^6$, and $p_j(x)=x^3 + \omega^5 x^2 + \omega^7 x + 1$, while $r_j = 0$ for all other irreducible factors of $x^{140} - 1$. According to Theorem \ref{Th-SO-DC-R}, $\mathcal{C}_6$ is $1$-Galois LCD, meaning that $\mathcal{C}_6 \cap \mathcal{C}_6^{\perp_1}=\{\mathbf{0}\}$ because 
$$\sum_{j=1}^{s} r_j \deg\left(p_j(x)\right)=5=\dim\left(\mathcal{C}_6\right).$$
\hfill $\diamond$
\end{example}

\section{Conclusion}
\label{conclusion}
This work began with the study of general linear codes, where we derived an explicit generator matrix for the intersection of any pair of linear codes. Then, we identified the largest reversible subcode contained in a given linear code. Our investigation then focused on the comprehensive class of MT codes We characterized the reversed code of an MT code, showing that it admits an MT structure and provided a GPM for it. We then examined the intersection of a pair of MT codes and showed that such intersection does not necessarily admit an MT structure. However, when it does, we proved that the blocks corresponding to indices where the shift constants differ are identically zero; this property was shown to hold for codes with minimum distance greater than the code index. When the two MT codes have identical shift constants and block lengths, we determined a GPM for their intersection. As a result, we established necessary and sufficient conditions for some properties, including code containment, trivial intersection, self-orthogonality, dual containment, LCD, and reversibility.

\section*{Acknowledgment}
This research was conducted at Université d'Artois, la Faculté des Sciences Jean Perrin, France, and was fully funded by the Science, Technology \& Innovation Funding Authority (STDF); International Cooperation Grants, project number 49294. Ramy Takieldin would like to express his deepest appreciation to Faculté des Sciences Jean Perrin for their hospitality and providing a fruitful research environment.


\bibliographystyle{plain}



\end{document}